\documentclass[11pt]{article}
\usepackage{amsfonts}
\usepackage{amsmath}
\usepackage{amssymb}
\usepackage{amsthm}
\usepackage{geometry}
\usepackage{rotating}
\usepackage[latin1]{inputenc}
\usepackage{array}
\usepackage[english]{babel}

\usepackage{bm}
\usepackage{natbib}

\usepackage{setspace}

\newcommand{\bfm}[1]   {\mbox{\boldmath{${#1}$}}}
\newcommand{\mvec}   {\mbox{\textnormal{vec}}}
\newcommand{\mvech}   {\mbox{\textnormal{vech}}}
\newcommand{\Cov}   {\mbox{\textnormal{Cov}}}
\newcommand{\mean}   {\mbox{\textnormal{E}}}

\theoremstyle{plain} \newtheorem{lemma}{Lemma}
\theoremstyle{plain} \newtheorem{theorem}{Theorem}
\theoremstyle{remark} 

\begin{document}
\title{Tests for multivariate normality based on canonical correlations}

\author
{ \bf{M{\aa}ns Thulin}$^{1}$}
\date{} 

\maketitle

\footnotetext[1]{Department of Mathematics, Uppsala University, P.O.Box 480, 751 06 Uppsala, Sweden.\\Phone: +46(0)184713389; Fax: +46(0)184713201; E-mail: thulin@math.uu.se}


\begin{abstract}
\noindent We propose new affine invariant tests for multivariate normality, based on independence characterizations of the sample moments of the normal distribution. The test statistics are obtained using canonical correlations between sets of sample moments, generalizing the Lin-Mudholkar test for normality. The tests are compared to some popular tests based on Mardia's skewness and kurtosis measures in an extensive simulation power study and are found to offer higher power against many of the alternatives.
\end{abstract}

\section{Introduction}

Many classical multivariate statistical methods are based on the assumption that the data comes from a multivariate normal distribution. Consequently, the use of such methods should be followed by an investigation of the assumption of normality. A number of tests for multivariate normality can be found in the literature, but the field has not been investigated to the same extent as have tests for univariate normality.

Let $\gamma=\mean(X-\mu)^3/\sigma^ 3$ denote the skewness of a univariate random variable $X$ and $\kappa=\mean(X-\mu)^ 4/\sigma^ 4-3$ denote its kurtosis. Both these quantities are 0 for the normal distribution but nonzero for many other distributions, and some common tests for univariate normality are therefore based on $\gamma$ and $\kappa$.

Different analog multivariate measures of skewness and kurtosis have been proposed, perhaps most notably by \citet{ma1}. Said measures have been used for various tests for multivariate normality in the last few decades. Some of these tests, in particular the tests that use Mardia's skewness and kurtosis measures as test statistics, have proved to have high power in many simulation studies \citep{he1,mm1,mm2} and new tests for normality based on multivariate skewness and kurtosis continue to be published today \citep{dh1,kth1}.

In this paper ten new tests for normality, all related to multivariate skewness or kurtosis, are proposed. Their common basis is independence characterizations of sample moments of the multivariate normal distribution.

The first five are based on the following well-known characterization: the i.i.d. variables $\bfm{X_1}, \ldots, \bfm{X_n}$ are normal if and only if the sample mean vector $\bfm{\bar{X}}=(\bar{X}_1,\bar{X}_2,\ldots,\bar{X}_p)'$ and the sample covariance matrix $\bfm{S}$ are independent.

When trying to use this characterization for tests for normality, it quickly becomes evident that it is difficult to test the independence of these two arrays. It is however, as we shall see, possible to study the covariance between the elements of the two arrays in a satisfactory manner.

Thus, assume that $\bfm{X_1}, \ldots, \bfm{X_n}$ are i.i.d. $p$-variate random variables with nonsingular covariance matrix $\bfm{\Sigma}$. Let $\bfm{\bar{X}}=(\bar{X}_1,\bar{X}_2,\ldots,\bar{X}_p)'$ be the sample mean vector and let 
\[
\bfm{S} =
\left[ \begin{array}{cccc}
S_{11} & S_{12} & \cdots & S_{1p} \\
S_{12} & S_{22} & \cdots & S_{2p} \\
\vdots & \vdots & \ddots & \vdots \\
S_{1p} & S_{2p} & \cdots & S_{pp} 
\end{array} \right]
\]
be the sample covariance matrix with $S_{ij}=(n-1)^{-1}\sum_{k=1}^n(X_{k,i}-\bar{X}_i)(X_{k,i}-\bar{X}_j)$. Define
\[
\bfm{u}=(S_{11},S_{12},\ldots,S_{1p},S_{22},S_{23},\ldots,S_{2p},S_{33},\ldots,S_{p-1,p},S_{pp})'
\]
so that $\bfm{u}$ is a vector containing the $q=p(p+1)/2$ distinct elements of $\bfm{S}$. Now, consider the covariance matrix of the vector $(\bfm{\bar{X}'},\bfm{u'})'$, in the following denoted $(\bfm{\bar{X}},\bfm{u})$:
\begin{equation}\label{covxu} \Cov((\bfm{\bar{X}},\bfm{u})) = \left[ \begin{array}{rr}
\bfm{\bfm{\Lambda}_{11}} & \bfm{\bfm{\Lambda}_{12}} \\
\bfm{\bfm{\Lambda}_{21}} & \bfm{\bfm{\Lambda}_{22}}\end{array} \right] \end{equation}
where $\bfm{\Lambda_{11}}=\Cov(\bfm{\bar{X}})$, $\bfm{\Lambda_{22}}=\Cov(\bfm{u})$, $\bfm{\Lambda_{21}}=\bfm{\bfm{\Lambda}_{12}^{'}}$ and $\bfm{\bfm{\Lambda}_{12}}$ contains covariances of the type $\Cov(\bar{X}_i, S_{jk})$, $i,j,k=1,\ldots,p$. We will use this matrix when we construct our tests.

Our main tool for doing so will be canonical correlations. This multivariate generalization of Pearson's correlation coefficient is defined as follows. The first canonical correlation between $\bfm{\bar{X}}$ and $\bfm{u}$ is the largest correlation coefficient between linear combinations of $\bfm{\bar{X}}$ and $\bfm{u}$:
\[
\lambda_1=\max_{\bfm{a}\in\mathbb{R}^p,\bfm{b}\in\mathbb{R}^q}|\rho(\bfm{a'\bar{X}},\bfm{b'u})|.
\]
The second canonical correlation $\lambda_2$ is the largest correlation coefficient between linear combinations of $\bfm{\bar{X}}$ and $\bfm{u}$ that are uncorrelated with the linear combinations corresponding to the first canonical correlation. For the third canonical correlation $\lambda_3$, the condition is that the linear combinations should be uncorrelated to the first two canonical correlations, and so on.

The canonical correlations can be obtained in different ways. We will use that $\lambda_1^2,\ldots,\lambda_p^2$ are the eigenvalues of the matrix
\begin{equation}\label{egenmatris}
\bfm{\Lambda_{11}}^{-1}\bfm{\Lambda_{12}}\bfm{\Lambda_{22}}^{-1}\bfm{\Lambda_{21}},
\end{equation}
see for instance \citet{ma3}.

Similarly to what was done above for the covariance, let
\[
S_{ijk}=\frac{n}{(n-1)(n-2)}\sum_{r=1}^n(X_{r,i}-\bar{X}_i)(X_{r,j}-\bar{X}_j)(X_{r,k}-\bar{X}_k)
\]
and
\[ \bfm{v}=(S_{111},S_{112},\ldots,S_{pp(p-1)},S_{ppp})',\]
a vector of length $p+p(p-1)+p(p-1)(p-2)/6$.
The five other tests are based on the fact that $\bfm{\bar{X}}$ and $\bfm{v}$ are independent if $\bfm{X_1}, \ldots, \bfm{X_n}$ are normal. The covariance matrix of $(\bfm{\bar{X}},\bfm{v})$ can be written as
\begin{equation}\label{covxv} \Cov((\bfm{\bar{X}},\bfm{v})) = \left[ \begin{array}{rr}
\bfm{\bfm{\Psi}_{11}} & \bfm{\bfm{\Psi}_{12}} \\
\bfm{\bfm{\Psi}_{21}} & \bfm{\bfm{\Psi}_{22}}\end{array} \right] \end{equation}
where $\bfm{\Psi_{11}}=\Cov(\bfm{\bar{X}})$, $\bfm{\Psi_{22}}=\Cov(\bfm{v})$, $\bfm{\Psi_{21}}=\bfm{\bfm{\Psi}_{12}^{'}}$ and $\bfm{\bfm{\Psi}_{12}}$ contains covariances of the type $\Cov(\bar{X}_i, S_{jkl})$, $i,j,k,l=1,\ldots,p$.

In Section \ref{covexakt} we state explicit expressions for the elements of $\Cov((\bfm{\bar{X}},\bfm{u}))$ and $\Cov((\bfm{\bar{X}},\bfm{v}))$ in terms of moments of $\bfm{X}=(X_1,\ldots,X_p)$. In Section \ref{skewtester} we reexamine Mardia's measure of multivariate skewness in the light of the results regarding $\Cov((\bfm{\bar{X}},\bfm{u}))$ and propose new tests for normality, all related to multivariate skewness. These can be viewed as multivariate generalizations of the $Z_2'$ modification \citep{er1} of the $Z_2$ test introduced by \citet{lm1}. In Section \ref{kurttester} we use the expressions for $\Cov((\bfm{\bar{X}},\bfm{v}))$ to construct new tests for normality, related to multivariate kurtosis. These, in turn, are generalizations of the $Z_3'$ modification of the $Z_3$ test proposed by \citet{mud2}. The results of a simulation study comparing the new tests with tests based on Mardia's skewness and kurtosis measures is presented in Section \ref{skevsim}. The text concludes with an appendix containing proofs and tables.

\section{Explicit expressions for the covariances}
\subsection{Elements of matrices}\label{covexakt}
In the following theorems we state explicit expressions for the elements of $\Cov((\bfm{\bar{X}},\bfm{u}))$ and $\Cov((\bfm{\bar{X}},\bfm{v}))$ in terms of moments of $(X_1,\ldots,X_p)$. These covariances can be obtained by tedious but routine calculations of the moments involved, that are much simplified by the use of tensor notation, as described in \citet{mcc2}. All five covariances can be found scattered in the literature, expressed using cumulants: (\ref{kovarians0})-(\ref{kovarianscov}) are all given in Section 4.2.3 of \citet{mcc2}, (\ref{psi12}) is found in Problem 4.5 of \cite{mcc2} and (\ref{psi22}) is expression (7) in \cite{ka1}.

\begin{theorem}\label{covthm}
Assume that $\bfm{X_1}, \ldots, \bfm{X_n}$ are i.i.d. $p$-variate random variables with\\ $\mean|X_iX_jX_kX_l|<\infty$ for $i,j,k,l=1,2,\ldots,p$. Let $\mu_{i_1,\ldots,i_s}=\mean (X_{i_1}-\mu_{i_1})(X_{i_2}-\mu_{i_2})\cdots(X_{i_s}-\mu_{i_s})$. Then, for $n\geq 2p+p(p-1)/2$ and $i,j,k,l=1,2,\ldots,p$
\begin{itemize}
\item[(i)] the elements of $\bfm{\Lambda}_{11}$ are
\begin{equation}\label{kovarians0}
\begin{split}
\Cov(\bar{X}_i, \bar{X}_j)&=\frac{1}{n}\mu_{ij},
\end{split}\end{equation}
\item[(ii)] the elements of $\bfm{\Lambda}_{12}$ and $\bfm{\Lambda}_{21}$ are
\begin{equation}\label{kovarians1}
\begin{split}
\Cov(\bar{X}_i, S_{jk})&=\frac{1}{n}\mu_{ijk}
\end{split}\end{equation}
and
\item[(iii)] the elements of $\bfm{\Lambda}_{22}$ are
\begin{equation}\label{kovarianscov}
\begin{split}
\Cov(S_{ij}, S_{kl})=&\frac{1}{n}(\mu_{ijkl}-\mu_{ij}\mu_{kl})+\frac{1}{n(n-1)}(\mu_{ik}\mu_{jl}+\mu_{il}\mu_{jk}).
\end{split}\end{equation}
\end{itemize}
\end{theorem}
\noindent
Since $\bfm{\Psi}_{11}=\bfm{\Lambda}_{11}$, we only give the expressions for $\bfm{\Psi_{22}}$ and $\bfm{\Psi_{12}}$ in the following theorem.

\begin{theorem}\label{covthm2}
Assume that $\bfm{X_1}, \ldots, \bfm{X_n}$ are i.i.d. $p$-variate random variables with\\ $\mean|X_\alpha X_\beta X_\gamma X_\delta X_\epsilon X_\zeta|<\infty$ for $\alpha,\ldots,\zeta=1,2,\ldots,p$. Let $\mu_{i_1,\ldots,i_s}=\mean (X_{i_1}-\mu_{i_1})(X_{i_2}-\mu_{i_2})\cdots(X_{i_s}-\mu_{i_s})$. Then, for $n\geq 2p+p(p-1)+p(p-1)(p-2)/6$ and $i,j,k,r,s,t=1,2,\ldots,p$
\begin{itemize}
\item[(i)] the elements of $\bfm{\Psi}_{12}$ and $\bfm{\Psi}_{21}$ are
\begin{equation}\label{psi12}
\begin{split}
\Cov(\bar{X}_i, S_{rst})&=\frac{1}{n}\Big{(}\mu_{irst}-\mu_{ir}\mu_{st}-\mu_{is}\mu_{rt}-\mu_{it}\mu_{rs}\Big{)}
\end{split}\end{equation}
and
\item[(ii)] the elements of $\bfm{\Psi}_{22}$ are
\begin{equation}\label{psi22}
\begin{split}
\Cov(S_{ijk}, S_{rst})=\frac{1}{n}\lambda_{ijkrst} +\frac{1}{n-1}\Big{(}\sum^9\mu_{ir}(\mu_{jkst}&-\sum^3\mu_{jk}\mu_{st})+\sum^9\mu_{ijr}\mu_{kst}\Big{)}\\
&+\frac{n}{(n-1)(n-2)}\sum^6\mu_{ir}\mu_{js}\mu_{kt}
\end{split}\end{equation}
where $\lambda_{ijkrst}$ is given below and $\sum^k$ denotes summation over $k$ distinct permutations of $i,j,k,r,s,t$. In particular, in $\sum^9\mu_{ir}(\ldots)$ the summation is taken over all permutations of $i,j,k,r,s,t$ where $i$ and either of $j,k$ switch places and/or $r$ and either of $s,t$ switch places. In $\sum^9\mu_{ijr}\mu_{kst}$ the summation is taken over all permutations except $\mu_{ijk}\mu_{rst}$. Finally, in $\sum^3\mu_{jk}\mu_{st}$ and
\[
\lambda_{ijkrst}=\mu_{ijkrst}-\sum^{15}\mu_{ij}(\mu_{krst}-\sum^3\mu_{kr}\mu_{st})-\sum^{10}\mu_{ijk}\mu_{rst}-\sum^{15}\mu_{ij}\mu_{kr}\mu_{st}
\]
the sums are taken over all distinct permutations.
\end{itemize}
\end{theorem}

\subsection{Covariances expressed using multivariate moments}\label{multmom}
To express the covariances above using matrices of multivariate moments, we will need some tools from matrix algebra, namely the Kronecker product $\otimes$, the $\mvec$ operator and a commutation matrix, all of which are described in \citet{rosen1}. See also \citet{kollo1}.

For a $p\times q$ matrix $\bfm{A}=\{a_{ij}\}$ and an $r\times s$ matrix $\bfm{B}$, the Kronecker product $\bfm{A}\otimes\bfm{B}$ is the $pr\times qs$ matrix $\{a_{ij}\bfm{B}\}$, $i=1,\ldots,p$, $j=1,\ldots,q$. The $\mvec$ operator stacks the columns of a matrix underneath eachother, forming a single vector. If the columns of the $p\times q$ matrix $\bfm{A}$ are denoted $\bfm{a_1},\ldots,\bfm{a_q}$ then $\mvec(\bfm{A})=(\bfm{a_1'},\ldots,\bfm{a_q'})'$ is a vector of length $pq$.

Let $\bfm{\mu}$ be the mean vector and $\bfm{\Sigma}$ be the covariance matrix of a random variable $\bfm{X}$. The third central moments of $\bfm{X}$ can be written as the $p\times p^2$ matrix
\[
\bar{m}_3(\bfm{X})=\mean\Big{\lbrack}(\bfm{X}-\bfm{\mu})\otimes(\bfm{X}-\bfm{\mu})'\otimes(\bfm{X}-\bfm{\mu}) \Big{\rbrack}'
\]
which contains all third order moments of $\bfm{X}$. Similarly, all fourth order moments of $\bfm{X}$ are found in the symmetric $p^2\times p^2$ matrix
\[
\bar{m}_4(\bfm{X})=\mean\Big{\lbrack}(\bfm{X}-\bfm{\mu})(\bfm{X}-\bfm{\mu})'\otimes(\bfm{X}-\bfm{\mu})(\bfm{X}-\bfm{\mu})' \Big{\rbrack},
\]
which we refer to as the fourth central moment of $\bfm{X}$. Finally, let $\bfm{K}_{pp}$ be the $p^2\times p^2$ commutation matrix, which consists of $p^2$ blocks of size $p\times p$, such that in the $ij$th block all elements equal 0 except for element $ji$, which equals 1. Thus, for instance, when $p=3$ the commutation matrix is
\[ \bfm{K}_{33}=\left( \begin{array}{ccc|ccc|ccc}
1&  0&  0&  0&  0&  0&  0&  0&  0\\
0&  0&  0&  1&  0&  0&  0&  0&  0\\
0&  0&  0&  0&  0&  0&  1&  0&  0\\\hline
0&  1&  0&  0&  0&  0&  0&  0&  0\\
0&  0&  0&  0&  1&  0&  0&  0&  0\\
0&  0&  0&  0&  0&  0&  0&  1&  0\\\hline
0&  0&  1&  0&  0&  0&  0&  0&  0\\
0&  0&  0&  0&  0&  1&  0&  0&  0\\
0&  0&  0&  0&  0&  0&  0&  0&  1 \end{array} \right).\] 
The third and fourth central moments can now be used to give explicit expression for the covariance matrices in Theorem \ref{covthm}. It turns out that it is helpful to use the vector
\[
\bfm{\nu}=\mvec(\bfm{S})=(S_{11},S_{12},\ldots,S_{1p},S_{12},S_{22},S_{23},\ldots,S_{p-1,p},S_{pp})'
\]
instead of $\bfm{u}$ when writing these expressions, bearing in mind that we wish to study the canonical correlations between $\bfm{\bar{X}}$ and $\bfm{u}$ and not the covariance matrices themselves. Those canonical correlations are the same as the canonical correlations between $\bfm{\bar{X}}$ and $\bfm{\nu}$. To see this, note that for every linear combination $\bfm{a'}\bfm{u}$ there exists a $\bfm{b}$ so that $\bfm{b'}\bfm{\nu}=\bfm{a'}\bfm{u}$ and therefore, by the definition of canonical correlations, the canonical correlations must coincide.

We partition $\Cov((\bfm{\bar{X}},\bfm{\nu}))$ in the same fashion as (\ref{covxu}):
\[ \Cov((\bfm{\bar{X}},\bfm{\nu})) = \left[ \begin{array}{rr}
\bfm{\Lambda}_{11}^{\nu} & \bfm{\Lambda}_{12}^{\nu} \\
\bfm{\Lambda}_{21}^{\nu} & \bfm{\Lambda}_{22}^{\nu}\end{array} \right]. \]
By direct computation, it can be seen that
\[\begin{split}
&\bfm{\Lambda}_{11}^{\nu}=\frac{1}{n}\bfm{\Sigma},\\
&\bfm{\Lambda}_{12}^{\nu}=\frac{1}{n}\bar{m}_3(\bfm{X})\qquad\mbox{ and}\\
&\bfm{\Lambda}_{22}^{\nu}=\frac{1}{n}\Big{(}\bar{m}_4(\bfm{X})-\mvec(\bfm{\Sigma})\mvec(\bfm{\Sigma})'\Big{)}+\frac{1}{n(n-1)}(\bfm{I}_{p^2}+\bfm{K}_{pp})(\bfm{S}\otimes\bfm{S}).
\end{split}\]
Although the above expressions provide a nice description of the covariance structure, they are of less interest when it comes to computational matters. The fact that $\bfm{\nu}$ contains the same elements more than once means that $\bfm{\Lambda}_{22}^{\nu}$ is singular and thus not invertible. Hence, we can't study the matrix ${\bfm{\Lambda}_{11}^{\nu}}^{-1}\bfm{\Lambda}_{12}{\bfm{\Lambda}_{22}^{\nu}}^{-1}\bfm{\Lambda}_{21}$, but must use (\ref{egenmatris}) instead. Similar, but less revealing, expressions can be obtained for $\bfm{\Psi_{22}}$ and $\bfm{\Psi_{12}}$. The same problem occurs in that case, meaning that the partitioning (\ref{covxv}) is of greater use.


\section{Tests based on $\bfm{\bar{X}}$ and $\bfm{u}$}\label{skewtester}
\subsection{Mardia's multivariate skewness measure $\bfm{\beta_{1,p}}$ revisited}\label{mardiatest}
\citet{ma1,ma2} noted that for univariate random variables, asymptotically $\rho(\bar{X}, S^2)\approx \frac{1}{\sqrt{2}}\gamma$ if $\kappa$ is assumed to be negligible. Based on this, he used $\Cov((\bfm{\bar{X}},\bfm{u}))$ to construct a multivariate skewness measure. Studying the canonical correlations between $\bfm{\bar{X}}$ and $\bfm{u}$ he proposed the measure
\[
\beta_{1,p}=2\sum_{i=1}^ p\lambda_i^2
\]
where $\lambda_1,\ldots,\lambda_p$ are the canonical correlations. This expression reduces to $2\rho(\bar{X}, S^2)^ 2\approx\gamma^2$ for univariate random variables.

From the theory of canonical correlations we have that $\lambda_1^ 2,\ldots,\lambda_p^ 2$ are the eigenvalues of $\bfm{\Lambda_{11}}^{-1}\bfm{\Lambda_{12}}\bfm{\Lambda_{22}}^{-1}\bfm{\Lambda_{21}}$ and thus
\[
\beta_{1,p}=2tr(\bfm{\Lambda_{11}}^{-1}\bfm{\Lambda_{12}}\bfm{\Lambda_{22}}^{-1}\bfm{\Lambda_{21}}).
\]
Taking these moments to order $n^{-1}$ Mardia showed that
\[
\beta_{1,p}\approx \mean\Big{(}(\bfm{X}-\bfm{\mu})'\bfm{\Sigma}^{-1}(\bfm{Y}-\bfm{\mu})\Big{)}^ 3
\]
where $\bfm{X}$ and $\bfm{Y}$ are independent and identical random vectors. The sample counterpart of the above expression,
\begin{equation}\label{b1p}
b_{1,p}=\frac{1}{n^2}\sum_{i,j=1}^n\Big{(}(\bfm{X_i}-\bfm{\bar{X}})'\bfm{S^{-1}}(\bfm{X_j}-\bfm{\bar{X}}\Big{)}^3,
\end{equation}
is commonly used as a measure for multivariate skewness and as a test statistic for a test for multivariate normality.

In Section 2.8 of \cite{mcc2} Mardia's approximation of $\beta_{1,p}$ is shown to be a natural generalization of $\gamma^2$. It is however not necessarily a good approximation of the canonical correlations between $\bfm{\bar{X}}$ and $\bfm{u}$. An important assumption underlying Mardia's skewness measure is that the fourth moments of the distribution are negligible. Seeing as $\gamma^2-2\leq\kappa$; see \citet{dm1}; this seems like a rather strong condition. \citet{er1} noted that
\[
\rho_2=\rho(\bar{X}, S^2)=\frac{\gamma}{\sqrt{\kappa+3-\frac{n-3}{n-1}}}
\]
and used $\hat{\rho}_2=Z_2'$, the sample moment version of this quantity, as a test statistic for a test for normality, proposing a test that is a modified version of the $Z_2$ test \citep{lm1}. This statistic is similar to the sample skewness $\hat{\gamma}$ (commonly known as $\sqrt{b_1}$), which is often used for univariate tests for normality. The $Z_2'$ and $\hat{\gamma}$ tests are both clearly sensitive to deviations from normality in the form of skewness, but the $Z_2'$ test has the additional benefit that it also takes the relationship between the skewness and the kurtosis into account. In Thulins's simulation power study, $Z_2'$ was more powerful than $\hat{\gamma}$  against most of the alternatives under study.  It is therefore of interest to use Mardia's approach without any approximations, so as to include the fourth moments as well, in the hope that this will render a more powerful test for normality. The explicit expressions for $\Cov(\bar{X}_i, S_{jk})$ and $\Cov(S_{ij}, S_{kl})$ given in Theorem \ref{covthm} allows us to study $\bfm{\Lambda_{11}}^{-1}\bfm{\Lambda_{12}}\bfm{\Lambda_{22}}^{-1}\bfm{\Lambda_{21}}$ without approximations and to construct new test statistics.

It should be noted that differences in index notation complicate the situation somewhat here. Mardia's skewness is denoted $b_{1,p}$, with 1 as its index, whereas the univariate correlation statistic $Z_2'$ has 2 as its index. When generalizing $Z_2'$ to the multivariate setting we will keep the index 2, hoping that it won't be confused with Mardia's kurtosis measure $b_{2,p}$.

The factor 2 in the expression
\[
\beta_{1,p}=2tr(\bfm{\Lambda_{11}}^{-1}\bfm{\Lambda_{12}}\bfm{\Lambda_{22}}^{-1}\bfm{\Lambda_{21}})
\]
makes no real sense if we do not assume negligible fourth moments. We will therefore omit it in the following and instead study the quantity
\[
tr(\bfm{\Lambda_{11}}^{-1}\bfm{\Lambda_{12}}\bfm{\Lambda_{22}}^{-1}\bfm{\Lambda_{21}}).
\]
Let $\bfm{S_{11}}$, $\bfm{S_{22}}$, $\bfm{S_{12}}$ and $\bfm{S_{21}}$ be the sample counterparts of $\bfm{\Lambda_{11}}$, $\bfm{\Lambda_{22}}$, $\bfm{\Lambda_{12}}$ and $\bfm{\Lambda_{21}}$, where $\mu_{i_1,\ldots,i_s}=\mean (X_{i_1}-\mu_{i_1})(X_{i_2}-\mu_{i_2})\ldots(X_{i_s}-\mu_{i_s})$ are estimated by the sample moments
\begin{equation}\label{muest}
m_{i_1,\ldots,i_s}=n^{-1}\sum_{k=1}^n (x_{k,i_1}-\bar{x}_{i_1})(x_{k,i_2}-\bar{x}_{i_2})\ldots(x_{k,i_s}-\bar{x}_{i_s}).
\end{equation}
The test statistic for the new test is
\begin{equation}\label{rho2p}
Z_{2,p}^{(HL)}=tr(\bfm{S_{11}}^{-1}\bfm{S_{12}}\bfm{S_{22}}^{-1}\bfm{S_{21}}).
\end{equation}
The null hypothesis of normality is rejected if $Z_{2,p}^{(HL)}$ is sufficiently large.

$Z_{2,1}^{(HL)}$ coincides with $Z_2'^2$ from \citet{er1} and is thus equivalent with the $|Z_2'|$ test presented there. $Z_{2,2}^{(HL)}$ is a polynomial of degree 10 in 13 moments and the full formula takes up more than two pages. It is however readily computed using a computer, as is $Z_{2,p}^{(HL)}$ for higher $p$.

\subsection{Other test statistics from the theory for canonical correlations}
Let $\bfm{Y}$ and $\bfm{Z}$ be normal random vectors with
\[ \Cov((\bfm{Y},\bfm{Z})) = \left[ \begin{array}{rr}
\bfm{\bfm{\Sigma}_{11}} & \bfm{\bfm{\Sigma}_{12}} \\
\bfm{\bfm{\Sigma}_{21}} & \bfm{\bfm{\Sigma}_{22}}\end{array} \right] \]
partitioned like (\ref{covxu}). Let $\bfm{\hat{\Sigma}_{11}}$, $\bfm{\hat{\Sigma}_{22}}$ and $\bfm{\hat{\Sigma}_{12}}=\bfm{\hat{\Sigma}_{21}'}$ be the sample covariance matrices and $\hat{\nu}_1^2,\ldots,\hat{\nu}_p^2$ be the eigenvalues of $\bfm{\hat{\Sigma}_{11}}^{-1}\bfm{\hat{\Sigma}_{12}}\bfm{\hat{\Sigma}_{22}}^{-1}\bfm{\hat{\Sigma}_{21}}$. In Section 10.3 of \citet{ksh1} the test statistic of the likelihood ratio test of $H_0: \bfm{\Sigma_{12}}=\bfm{0}$ versus $H_1: \bfm{\Sigma_{12}}\neq\bfm{0}$ is shown to be
\begin{equation}\label{mlcc}
-n\log\prod_{i=1}^p(1-\hat{\nu}_i^2).
\end{equation}
Now, let $\hat{\lambda}_1^2\geq\hat{\lambda}_2^2\geq\ldots\geq\hat{\lambda}_p^2$ be the eigenvalues of $\bfm{S_{11}}^{-1}\bfm{S_{12}}\bfm{S_{22}}^{-1}\bfm{S_{21}}$. Assuming that the necessary moments exist, $\bfm{\bar{X}}$ and $\bfm{u}$ are asymptotically normal. Although $\bfm{S_{22}}$ and $\bfm{S_{12}}$ are not the usual sample covariance matrices, in the light of (\ref{mlcc}), this suggests the use of the following statistic for a test for normality: 
\begin{equation}\label{r2p}
Z_{2,p}^{(W)}=\prod_{i=1}^p(1-\hat{\lambda}_i^2).
\end{equation}
The null hypothesis of normality is rejected if $Z_{2,p}^{(W)}$ is sufficiently small.

Another quantity that has been considered for a test of $H_0: \bfm{\Sigma_{12}}=\bfm{0}$, for instance by \citet{ba1}, is
\begin{equation}\label{omega2p}
Z_{2,p}^{(PB)}=\sum_{i=1}^p\frac{\hat{\lambda}_i^2}{1-\hat{\lambda}_i^2}.
\end{equation}
$Z_{2,p}^{(PB)}$ is similar to $Z_{2,p}^{(HL)}$, but weighs the correlation coefficients so that larger coefficients become more influential. The null hypothesis should be rejected for large values of $Z_{2,p}^{(PB)}$.

Finally, we can consider the statistics
\begin{equation}
Z_{2,p}^{(max)}=\max(\hat{\lambda}_1^2,\ldots,\hat{\lambda}_p^2)=\hat{\lambda}_1^2,
\end{equation}
and
\begin{equation}
Z_{2,p}^{(min)}=\min(\hat{\lambda}_1^2,\ldots,\hat{\lambda}_p^2)=\hat{\lambda}_p^2,
\end{equation}
large values of which imply non-normality. $Z_{2,p}^{(max)}$ seems perhaps like the most natural choice for a test statistic, as $\lambda_1=0$ implies that all canonical correlations are 0.

The statistics $Z_{2,p}^{(HL)}$, $Z_{2,p}^{(W)}$, $Z_{2,p}^{(PB)}$, $Z_{2,p}^{(max)}$ and $Z_{2,p}^{(min)}$ are all related to well-known statistics from multivariate analysis of variance; they are analogs to the Hotelling-Lawley trace, Wilk's $\Lambda$, the Pillai-Bartlett trace and Roy's greatest and smallest root, respectively. For $p=1$ these statistics are all equivalent to the $|Z_2'|$ test from \citet{er1}.

\subsection{Theoretical results}
Some fundamental properties of the new test statistics are presented in the following theorem. Its proof is given in the Appendix.

\begin{theorem}\label{rho2thm}
Assume that $\bfm{X_1}, \ldots, \bfm{X_n}$ are i.i.d. $p$-variate random variables fulfilling the conditions of Theorem \ref{covthm}. Then, for $n\geq 2p+p(p-1)/2$ and $i,j,k=1,2,\ldots,p$
\begin{itemize}
\item[(i)] $Z_{2,p}^{(HL)}$, $Z_{2,p}^{(W)}$, $Z_{2,p}^{(PB)}$ $Z_{2,p}^{(max)}$ and $Z_{2,p}^{(min)}$ are affine invariant, i.e. invariant under nonsingular linear transformations $\bfm{AX}+\bfm{b}$,
\item[(ii)] The population canonical correlation $\lambda_1=\max_{\bfm{a},\bfm{b}}|\rho(\bfm{a\bar{X}},\bfm{bu})|=0$ if $\mu_{ijk}= 0$ for all $i,j,k$ and $>0$ if $\mu_{ijk}\neq 0$ for at least one combination of $i,j,k$, and
\item[(iii)] $Z_{2,p}^{(HL)}$, $Z_{2,p}^{(W)}$, $Z_{2,p}^{(PB)}$ $Z_{2,p}^{(max)}$ and $Z_{2,p}^{(min)}$ converge almost surely to the corresponding functions of the population canonical correlations $\lambda_1\geq\lambda_2\geq\ldots\geq\lambda_p$.
\end{itemize}
\end{theorem}

Since the statistics are affine invariant, their distributions are the same for all $p$-variate normal distributions for a given sample size $n$. These null distributions are easily obtained using Monte Carlo simulation.

Since $\lambda_1\geq\lambda_j$ for $j>1$, $\lambda_1=0$ implies that all population canonical correlations are 0, as is the case for the normal distribution. The tests should therefore not be sensitive to distributions with that kind of symmetry. All five statistics are, by (ii) and (iii), however consistent against alternatives where $\mu_{ijk}\neq 0$ for at least one combination of $i,j,k$. In particular, they are sensitive to alternatives with skew marginal distributions.


\section{Tests based on $\bfm{\bar{X}}$ and $\bfm{v}$}\label{kurttester}
\subsection{Test statistics}
\citet{ma1,ma2} proposed the multivariate kurtosis measure
\[
\beta_{2,p}=\mean\Big{(}(\bfm{X}-\bfm{\mu})'\bfm{\Sigma}^{-1}(\bfm{Y}-\bfm{\mu})\Big{)}^2
\]
with sample counterpart
\begin{equation}\label{mardiab2p}
b_{2,p}=\frac{1}{n}\sum_{i=1}^n\Big{(}(\bfm{X_i}-\bfm{\bar{X}})'\bfm{S^{-1}}(\bfm{X_i}-\bfm{\bar{X}})\Big{)}^2.
\end{equation}

In the univariate setting
\begin{equation}\label{rho3grej}
\rho_3=\rho\Big{(}\bar{X},\frac{n}{(n-1)(n-2)}\sum_{i=1}^n(X_i-\bar{X})^3\Big{)}=\frac{\kappa}{\sqrt{\lambda+9\frac{n}{n-1}(\kappa+\gamma^2)+\frac{6n^2}{(n-1)(n-2)}}},
\end{equation}
where $\lambda=\frac{\mu_6}{\sigma^6}-15\kappa-10\gamma^2-15$ is the sixth standardized cumulant \citep{er1}. In a simulation power study he found the test for normality based on $\hat{\rho}_3=Z_3'$, the sample counterpart of (\ref{rho3grej}), to have a better overall performance than the popular $\hat{\kappa}=b_2=b_{2,1}$ test.

The ideas used in Section \ref{skewtester} for $\Cov((\bfm{\bar{X}},\bfm{u}))$ can also be used for $\Cov((\bfm{\bar{X}},\bfm{v}))$ in an analog manner, yielding multivariate generalizations of (\ref{rho3grej}). This leads to five new tests for normality, as described below.

Let $\bfm{P_{11}}$, $\bfm{P_{22}}$, $\bfm{P_{12}}$ and $\bfm{P_{21}}$ be the sample counterparts of $\bfm{\Psi_{11}}$, $\bfm{\Psi_{22}}$, $\bfm{\Psi_{12}}$ and $\bfm{\Psi_{21}}$, where the $\mu_{i_1,\ldots,i_s}$ are estimated by the sample moments, as in (\ref{muest}) above. Let $\hat{\psi}_1^2\geq\ldots\geq\hat{\psi}_p^2$ be the eigenvalues of $\bfm{P_{11}}^{-1}\bfm{P_{12}}\bfm{P_{22}}^{-1}\bfm{P_{21}}$.

The test statistics for the new tests are
\begin{equation}\label{rho3p}
Z_{3,p}^{(HL)}=tr(\bfm{P_{11}}^{-1}\bfm{P_{12}}\bfm{P_{22}}^{-1}\bfm{P_{21}})=\sum_{i=1}^p\hat{\psi}_i^2,
\end{equation}
\begin{equation}\label{r3p}
Z_{3,p}^{(W)}=\prod_{i=1}^p(1-\hat{\psi}_i^2),
\end{equation}
\begin{equation}\label{omega3p}
Z_{3,p}^{(PB)}=\sum_{i=1}^p\frac{\hat{\psi}_i^2}{1-\hat{\psi}_i^2},
\end{equation}
\begin{equation}
Z_{3,p}^{(max)}=\max(\hat{\psi}_1^2,\ldots,\hat{\psi}_p^2)=\hat{\psi}_1^2
\end{equation}
and
\begin{equation}
Z_{3,p}^{(min)}=\min(\hat{\psi}_1^2,\ldots,\hat{\psi}_p^2)=\hat{\psi}_p^2.
\end{equation}
Large values of $Z_{3,p}^{(HL)}$, $Z_{3,p}^{(PB)}$ $Z_{3,p}^{(max)}$ and $Z_{3,p}^{(min)}$ and small values of $Z_{3,p}^{(W)}$ imply non-normality. All five statistics are equivalent to $|Z_3'|$ from \citet{er1} for $p=1$.

\subsection{Theoretical results}
The following theorem mimics Theorem \ref{rho2thm} above. Its proof is given in the Appendix.

\begin{theorem}\label{rho3thm}
Assume that $\bfm{X_1}, \ldots, \bfm{X_n}$ are i.i.d. $p$-variate random variables fulfilling the conditions of Theorem \ref{covthm2}. Then, for $n\geq 2p+p(p-1)+p(p-1)(p-2)/6$ and $i,j,k,r,s,t=1,2,\ldots,p$
\begin{itemize}
\item[(i)] $Z_{3,p}^{(HL)}$, $Z_{3,p}^{(W)}$, $Z_{3,p}^{(PB)}$, $Z_{3,p}^{(max)}$ and $Z_{3,p}^{(min)}$ are affine invariant, i.e. invariant under nonsingular linear transformations $\bfm{AX}+\bfm{b}$,
\item[(ii)]The population canonical correlation $\psi_1=\max_{\bfm{a},\bfm{b}}|\rho(\bfm{a\bar{X}},\bfm{bv})|=0$ if $\mu_{irst}-\mu_{ir}\mu_{st}-\mu_{is}\mu_{rt}-\mu_{it}\mu_{rs}=0$ for all $i,r,s,t=1,\ldots,p$ and $>0$ otherwise, and
\item[(iii)]  $Z_{3,p}^{(HL)}$, $Z_{3,p}^{(W)}$, $Z_{3,p}^{(PB)}$ and $Z_{3,p}^{(max)}$ converge almost surely to the corresponding functions of the population canonical correlations $\psi_1\geq\psi_2\geq\ldots\geq\psi_p$.
\end{itemize}
\end{theorem}

Using the affine invariance, the null distributions of the statistics can be obtained through Monte Carlo simulation. 

By (ii) and (iii) all five statistics are consistent against alternatives where $\mu_{irst}-\mu_{ir}\mu_{st}-\mu_{is}\mu_{rt}-\mu_{it}\mu_{rs}\neq 0$ for at least one combination of $i,r, s, t$.

\section{Simulation results}\label{skevsim}
\subsection{The simulation study}
To evaluate the performance of the new $Z_{2,p}$ and $Z_{3,p}$ tests, a Monte Carlo study of their power was carried out. The tests were compared to the test based on Mardia's skewness measure $b_{1,p}$ (\ref{b1p}), the test based on Mardia's kurtosis measure $b_{2,p}$ (\ref{mardiab2p}) and the Mardia-Kent omnibus test $T$ \citep{mk1}, in which the skewness and kurtosis measures are combined. The tests were compared for $n=20$ and $n=50$ for $p=2$ and $p=3$.

In many power studies for multivariate tests for normality alternatives with independent marginal distributions have been used. We believe that this can be misleading, as distributions with independent marginals are uncommon in practice and indeed of little interest in the multivariate setting, where the dependence structure of the marginals often is paramount. For this reason, we decided to focus mainly on alternatives with a more complex dependence structure in our study. One alternative with independent exponential marginals, that have been used in many previous power studies, is included for reference.

The alternatives used in the study are presented in Tables \ref{tabalt2}-\ref{tabalt3}. The asymmetric multivariate Laplace distribution mentioned in Table \ref{tabalt3} is described in \citet{kotz1}.

In order to see which alternatives that the different tests could be sensitive to, the population values of the statistics were determined for all alternatives. For most distributions the values were computed numerically, to one decimal place for Mardia's statistics and to two decimal places for the $Z_{2,p}$ and $Z_{3,p}$ tests. The population values are given in Table \ref{altpop} in the Appendix.


\begin{table}[ht]
\begin{center}
%
\caption{Alternatives constructed using their marginals.}\label{tabalt2}

\begin{tabular}{p{4.5cm} | p{8cm}  }
\textbf{Distr. of $Y_i$} & \textbf{Construction of $(Y_1,\ldots,Y_p)$}\\\hline
Indep. $Exp(1)$ & $Y_1,\ldots,Y_p$ i.i.d. $Exp(1)$.\\
$LogN(0,2)$ & $X_0\sim LogN(0,1)$ indep. of $X_1,\ldots,X_p$ i.i.d. $LogN(0,1)$. $Y_i=X_0\cdot X_p$.\\
$LogN(0,1)$ & $X_0\sim LogN(0,0.5)$ indep. of $X_1,\ldots,X_p$ i.i.d. $LogN(0,0.5)$. $Y_i=X_0\cdot X_p$.\\
$LogN(0,0.5)$ & $X_0\sim LogN(0,0.25)$ indep. of $X_1,\ldots,X_p$ i.i.d. $LogN(0,0.25)$. $Y_i=X_0\cdot X_p$.\\
$Laplace(0,1)$ (type I) & $X_0\sim Exp(1)$ indep. of $X_1,\ldots,X_p$ i.i.d. $Exp(1)$. $Y_i=X_i-X_0$.\\
$Laplace(0,1)$ (type II) & $X_0\sim N(0,1)$ indep. of $X_{i,1},X_{i,2},X_{i,3}$, $i=1,\ldots,p$ i.i.d. $N(0,1)$. $Y_i=X_0\cdot X_{i,1}+X_{i,2}\cdot X_{i,3}$.\\
$Beta(1,1)$ & $X_0\sim \Gamma(1,1)$ indep. of $X_1,\ldots,X_p$ i.i.d. $\Gamma(1,1)$. $Y_i=X_i/(X_i+X_0)$.\\
$Beta(1,2)$ & $X_0\sim \Gamma(2,1)$ indep. of $X_1,\ldots,X_p$ i.i.d. $\Gamma(1,1)$. $Y_i=X_i/(X_i+X_0)$.\\
$Beta(2,2)$ & $X_0\sim \Gamma(2,1)$ indep. of $X_1,\ldots,X_p$ i.i.d. $\Gamma(2,1)$. $Y_i=X_i/(X_i+X_0)$.\\
$\chi_2^2$ & $X_0\sim \Gamma(0.5,0.5)$ indep. of $X_1,\ldots,X_p$ i.i.d. $\Gamma(0.5,0.5)$. $Y_i=X_0+X_i$.\\
$\chi_8^2$ & $X_0\sim \Gamma(2,0.5)$ indep. of $X_1,\ldots,X_p$ i.i.d. $\Gamma(2,0.5)$. $Y_i=X_0+X_i$.\\\hline
\end{tabular}
\end{center}
\end{table}

Using R, the thirteen tests were applied to 1,000,000 samples from each alternative and each combination of $n$ and $p$. The null distributions for all test statistics were estimated using 100,000 normal samples.


\begin{table}[ht]
\begin{center}
\caption{Purely multivariate alternatives. Here $\bfm{\Sigma_r}$ is a covariance matrix with unit variances and correlations $r$.}\label{tabalt3}

\begin{tabular}{l | l  }
\textbf{Distribution} & \textbf{Description}\\\hline
$t(2)$ & Multivariate $t$ distribution, symmetric\\
$AL(\mathbf{0},\bfm{\Sigma_{0}})$& Symmetric multivariate Laplace\\
$AL(\mathbf{1},\bfm{\Sigma_{0}})$& Asymmetric multivariate Laplace\\
$AL(\mathbf{3},\bfm{\Sigma_{0}})$& Asymmetric multivariate Laplace\\
$AL(\mathbf{1},\bfm{\Sigma_{0.5}})$& Asymmetric multivariate Laplace\\
$AL(\mathbf{1},\bfm{\Sigma_{0.9}})$& Asymmetric multivariate Laplace\\
$\frac{9}{10}N(\bfm{0},\bfm{\Sigma_{0}})+\frac{1}{10}N(\bfm{1},\bfm{\Sigma_{0}})$ & Location polluted normal mixture\\
$\frac{9}{10}N(\bfm{0},\bfm{\Sigma_{0}})+\frac{1}{10}N(\bfm{2},\bfm{\Sigma_{0}})$ & Location polluted normal mixture\\
$\frac{9}{10}N(\bfm{0},\bfm{\Sigma_{0}})+\frac{1}{10}N(\bfm{0},\bfm{\Sigma_{0.5}})$ & Rotation polluted normal mixture\\
$\frac{9}{10}N(\bfm{0},\bfm{\Sigma_{0}})+\frac{1}{10}N(\bfm{1},\bfm{\Sigma_{0.5}})$ & Scale and rotation polluted normal mixture\\
$\frac{9}{10}N(\bfm{0},\bfm{\Sigma_{0}})+\frac{1}{10}N(\bfm{2},\bfm{\Sigma_{0.5}})$ & Scale and rotation polluted normal mixture\\
$\frac{3}{4}N(\bfm{0},\bfm{\Sigma_{0}})+\frac{1}{4}N(\bfm{1},\bfm{\Sigma_{0}})$ & Heavily location polluted normal mixture\\
$\frac{3}{4}N(\bfm{0},\bfm{\Sigma_{0}})+\frac{1}{4}N(\bfm{2},\bfm{\Sigma_{0}})$ & Heavily location polluted normal mixture\\
$\frac{3}{4}N(\bfm{0},\bfm{\Sigma_{0}})+\frac{1}{4}N(\bfm{0},\bfm{\Sigma_{0.5}})$ & Heavily rotation polluted normal mixture\\
$\frac{3}{4}N(\bfm{0},\bfm{\Sigma_{0}})+\frac{1}{4}N(\bfm{1},\bfm{\Sigma_{0.5}})$ & Heavily scale and rotation polluted normal mixture\\
$\frac{3}{4}N(\bfm{0},\bfm{\Sigma_{0}})+\frac{1}{4}N(\bfm{2},\bfm{\Sigma_{0.5}})$ & Heavily scale and rotation polluted normal mixture\\\hline
\end{tabular}
\end{center}
\end{table}

\subsection{Results and recommendations}
The results from the simulation study are presented in Tables \ref{tab2}-\ref{tab4} in the Appendix.

Some of the results in the tables highlight the fact that what holds true for one combination of $p$ and $n$ can be false for a different combination. For instance, when $p=2$, $Z_{2,p}^{(max)}$ had higher power than $b_{1,p}$ for the $AL(\mathbf{1},\bfm{\Sigma_{0}})$ and the multivariate $\chi^2_8$ alternatives, but when $p=3$, $Z_{2,p}^{(max)}$ had lower power than $b_{1,p}$. This phenomenon merits further investigation, as it implies that power studies performed for low values of $p$ can be misleading when choosing between tests to use for higher-dimensional data. Similarly, when $p=2$, $b_{1,p}$ has higher power than $Z_{2,p}^{(max)}$ against the $LogN(0,1)$ distribution when $n=20$ but lower power when $n=50$.

For $p=2$, the $Z_{2,p}^{(max)}$ had the best overall performance against asymmetric alternatives, while Mardias skewness test $b_{1,p}$ and the $Z_{2,p}^{(W)}$ and $Z_{2,p}^{(PB)}$ tests also displayed a good average performance. For $p=3$ the performance of $Z_{2,p}^{(max)}$ was somewhat worse, whereas $b_{1,p}$, $Z_{2,p}^{(W)}$ and $Z_{2,p}^{(PB)}$ still showed good performance.

Looking at the normal mixtures, which can be viewed as contamined normal distributions, we see that $Z_{2,p}^{(max)}$ and $b_{1,p}$ were on a par for the mildy polluted mixtures (with a 9:1 mixing ratio) and that $Z_{2,p}^{(max)}$ in general had higher power for the heavily polluted mixtures (with a 3:1 mixing ratio). This suggests the use of the $Z_{2,p}^{(max)}$ statistic for a test for outliers, an idea that perhaps could be investigated further.

Among the symmetric alternatives, Mardias kurtosis test $b_{2,p}$ and the $Z_{3,p}$ tests were somewhat surprisingly outperformed by the skewness test and the $Z_{2,p}$ tests for some of the alternatives. The $Z_{3,p}$ were seen to be remarkably insensitive to some alternatives, both symmetric and asymmetric, and offered lower power than the $Z_{2,p}$ against the (symmetric) multivariate t-distribution. On the other hand, the $Z_{3,p}$ tests had significantly higher power than the other tests against the symmetric distributions with dependent short-tailed $Beta(1,1)$ (uniform) and $Beta(2,2)$ marginals.

Based on the simulation results, our recommendations are that the $Z_{2,p}^{(max)}$ test should be used against asymmetric alternatives when $p=2$. For higher $p$, $b_{1,p}$, $Z_{2,p}^{(W)}$ or $Z_{2,p}^{(PB)}$ should be used instead. In general, Mardia's $b_{2,p}$ test should be used against symmetric alternatives. For short-tailed symmetric alternatives, one of the $Z_{3,p}$ tests would be a better choice.

\subsubsection*{Acknowledgments}
The author wishes to thank Silvelyn Zwanzig for several helpful suggestions.

\section*{Appendix: proofs and tables}\label{appendix}
For the proof of Theorems \ref{rho2thm} and Theorems \ref{rho3thm} we need some basic properties of the Kronecker product $\otimes$ and $\mvech$ and $\mvec$ operators from \citet{hen1}. The basics of the Kronecker product and the $\mvec$ operator were described in Section \ref{multmom}.

We will use that
\[
(\bfm{A}\otimes\bfm{B})(\bfm{C}\otimes\bfm{D})=\bfm{AC}\otimes\bfm{BD},\qquad (\bfm{A}\otimes\bfm{B})'=\bfm{A'}\otimes\bfm{B'}
\]
and that if $\bfm{A}$ is a $p\times p$ matrix and $\bfm{B}$ a $q\times q$ matrix,
\[
\det(\bfm{A}\otimes\bfm{B})=\det(\bfm{A})^q\det(\bfm{A})^p.
\]
The $\mvech$ operator works as the $\mvec$ operator, except that it only contains each distinct element of the matrix once. For a symmetric matrix $\bfm{A}$, $\mvech(\bfm{A})$ thus contains only the diagonal and the elements above the diagonal, whereas $\mvec(\bfm{A})$ contains the diagonal elements and the off-diagonal elements twice.

We have the following relationship between the $\mvec$ operator and the Kronecker product:
\[
\mvec(\bfm{ABC})=(\bfm{C'}\otimes\bfm{A})\mvec(\bfm{B}).
\]
Furthermore, for a given symmetric $p\times p$ matrix $\bfm{A}$ there exists a $p(p+1)/2\times p^2$ matrix $\bfm{H}$ and a $p^2\times p(p+1)/2$ matrix $\bfm{G}$ such that
\[
\mvech(\bfm{A})=\bfm{H}\mvec(\bfm{A})\qquad\mbox{ and }\qquad\mvec(\bfm{A})=\bfm{G}\mvech(\bfm{A}).
\]

As a preparation for the proof of Theorem \ref{rho2thm}, we prove the following auxiliary lemma.

\begin{lemma}\label{auxlemma}
Assume that $\bfm{X},\bfm{X_1}, \ldots, \bfm{X_n}$ are i.i.d. $p$-variate random variables fulfilling the conditions of Theorem \ref{covthm}. Let $S_{ij}=(n-1)^{-1}\sum_{k=1}^n(X_{k,i}-\bar{X}_i)(X_{k,j}-\bar{X}_j)$ be the elements of the sample covariance matrix $\bfm{S}$.
\[
\bfm{u_X}=(S_{11},S_{12},\ldots,S_{1p},S_{22},S_{23},\ldots,S_{2p},S_{33},\ldots,S_{p-1,p},S_{pp})'
\]
is a vector with $q=p(p+1)/2$ distinct elements. Denote its covariance matrix $\Cov(\bfm{u_X})=\bfm{\Lambda_{22}}$.

Let $\bfm{A}$ be a nonsingular $p\times p$ matrix and let $\bfm{b}$ be a $p$-dimensional vector. Then there exists a nonsingular $q\times q$ matrix $\bfm{D}$ such that
\begin{itemize}
\item[(i)] the sample variances and covariances of $\bfm{Y}=\bfm{AX}+\bfm{b}$ are given by $\bfm{u_Y}=\bfm{Du_X}$,
\item[(ii)] $\Cov(\bfm{u_Y})=\bfm{D\Lambda_{22}D'}$ and
\item[(iii)] $\det(\bfm{D})=\det(\bfm{A})^{p+1}$,
\end{itemize}
\end{lemma}
\begin{proof}
First we note that $\bfm{u_X}=\mvech(\bfm{S})$.

The transformed sample $\bfm{AX}+\bfm{b}$ has sample covariance matrix $\bfm{ASA'}$, so we wish to study $\mvech(\bfm{ASA'})$. We have
\[
\mvec(\bfm{ASA'})=(\bfm{A}\otimes \bfm{A})\mvec(\bfm{S}).
\]
Moreover, since $\bfm{S}$ is symmetric there exist nonsingular matrices $\bfm{G}$ and $\bfm{H}$ such that
\[
\mvec(\bfm{S})=\bfm{G}\mvech(\bfm{S})\qquad\mbox{ and }\qquad \mvech(\bfm{S})=\bfm{H}\mvec(\bfm{S}).
\]
Thus
\[
\bfm{u_Y}=\mvech(\bfm{ASA'})=\bfm{H}(\bfm{A}\otimes\bfm{A})\bfm{G}\mvech({\bfm{S}})=:\bfm{D}\bfm{u_X},
\]
which establishes the existence of $\bfm{D}$. From Section 4.2 of \citet{hen1} we have
\[
\det(\bfm{D})=\det(\bfm{H}(\bfm{A}\otimes\bfm{A})\bfm{G})=\det(\bfm{A})^{p+1}
\]
which in nonzero, since $\bfm{A}$ is nonsingular. $\bfm{D}$ is hence also nonsingular. In conclusion, we have established the existence and nonsingularity of $\bfm{D}$ as well as (i) and (iii). Finally, (ii) follows immediately from (i).
\end{proof}
%
We now have the tools necessary to tackle Theorem \ref{rho2thm}.

\begin{proof}[Proof of Theorem \ref{rho2thm}]
\begin{itemize}
\item[(i)] From Theorem 10.2.4 in \citet{ma3} we have that the canonical correlations between the random vectors $\bfm{Y}$ and $\bfm{Z}$ are invariant under the nonsingular linear transformations $\bfm{AY}+\bfm{b}$ and $\bfm{CZ}+\bfm{d}$. Clearly all five statistics are invariant under changes in location, since $\bfm{S_{11}}$, $\bfm{S_{22}}$, $\bfm{S_{12}}$ and $\bfm{S_{21}}$ all share that invariance property. It therefore suffices to show that the nonsingular linear transformation $\bfm{AX}$ induces nonsingular linear transformations $\bfm{C\bar{X}}$ and $\bfm{Du}$. $\bfm{C}=\bfm{A}$ is immediate and the existence of $\bfm{D}$ is given by Lemma \ref{auxlemma}.
\item[(ii)] By part (ii) of Theorem \ref{covthm}, $\mu_{ijk}=0$ for all $i,j,k$ implies that $\bfm{\Lambda}_{12}=\bfm{0}$. But then $\bfm{\Lambda_{11}}^{-1}\bfm{\Lambda_{12}}\bfm{\Lambda_{22}}^{-1}\bfm{\Lambda_{21}}=\bfm{0}$ and all canonical correlations are 0. If $\mu_{ijk}\neq 0$ then $\rho(\bar{X}_i,S_{jk})\neq 0$. Thus the linear combinations $\bfm{a'\bar{X}}=\bar{X}_i$ and $\bfm{b'u}=S_{jk}$ have nonzero correlation. $\lambda_1$ must therefore be greater than 0.
\item[(iii)] Follows from the fact that the statistics are continuous function of sample moments that converge almost surely.
\end{itemize}
\end{proof}

The proofs of parts (ii) and (iii) of Theorem \ref{rho3thm} are analog to the previous proof. The proof for part (i) is however slightly different as we omit to explicitly give a matrix that gives a nonsingular linear transformation of $\bfm{v_X}$.

\begin{proof}[Proof of Theorem \ref{rho3thm} (i)]
In the spirit of Section \ref{multmom}, let the third order central moment of a multivariate random variable $\bfm{Z}$ be
\[\begin{split}
\bar{m}_3(\bfm{Z})&=\mean\Big{\lbrack}(\bfm{Z}-\mean\bfm{Z})\otimes(\bfm{Z}-\mean\bfm{Z})'\otimes(\bfm{Z}-\mean\bfm{Z}) \Big{\rbrack}'\\
&=\mean\Big{\lbrack}(\bfm{Z}-\mean\bfm{Z})\Big{(}(\bfm{Z}-\mean\bfm{Z})\otimes(\bfm{Z}-\mean\bfm{Z})\Big{)}'\Big{\rbrack}.
\end{split}\]
Given a sample $\bfm{X}_1,\ldots,\bfm{X}_p$, let $S_{ijk}=\frac{n}{(n-1)(n-2)}\sum_{r=1}^n(X_{r,i}-\bar{X}_i)(X_{r,j}-\bar{X}_j)(X_{r,k}-\bar{X}_k)$. When the distribution of $\bfm{Z}$ is the empirical distribution of said sample,
\[
\bfm{v_X}=(S_{111},S_{112},\ldots,S_{pp(p-1)},S_{ppp})'=\frac{n^2}{(n-1)(n-2)}\mvech\Big{(}\bar{m}_3(\bfm{Z})\Big{)}.
\]
Similarly $\mvec\Big{(}\bar{m}_3(\bfm{Z})\Big{)}$ stacks the elements of $\bar{m}_3(\bfm{Z})$ in a vector that simply is $\mvech\Big{(}\bar{m}_3(\bfm{Z})\Big{)}$ with a few repetitions:
\[
\bfm{w_X}=(S_{111},S_{112},\ldots,S_{112}\ldots,S_{pp(p-1)},S_{ppp})'=\frac{n^2}{(n-1)(n-2)}\mvec\Big{(}\bar{m}_3(\bfm{Z})\Big{)}.
\]
Thus, for each linear combination $\bfm{a'}\bfm{w_X}$ there exists a $\bfm{b}$ so that $\bfm{b'}\bfm{v_X}=\bfm{a'}\bfm{w_X}$ and therefore, by the definition of canonical correlations, the (sample) canonical correlations between $\bfm{\bar{X}}$ and $\bfm{v_X}$ are the same as those between $\bfm{\bar{X}}$ and $\bfm{w_X}$.

Writing $\bfm{Y}=\bfm{Z}-\mean\bfm{Z}$, we have $\bar{m}_3(\bfm{Z})=\mean\Big{(}\bfm{Y}(\bfm{Y}\otimes\bfm{Y})'\Big{)}$ and
\[\begin{split}
\bar{m}_3(\bfm{AZ})&=\mean\Big{(}\bfm{AY}(\bfm{AY}\otimes\bfm{AY})'\Big{)}=\mean\Big{(}\bfm{AY}(\bfm{Y}\otimes\bfm{Y})'(\bfm{A}\otimes\bfm{A})'\Big{)}\\
&=\bfm{A}\bar{m}_3(\bfm{Z})(\bfm{A}\otimes\bfm{A})'.
\end{split}\]
Hence
\[
\mvec\Big{(}\bar{m}_3(\bfm{AZ})\Big{)}=(\bfm{A}\otimes\bfm{A}\otimes\bfm{A})\mvec\Big{(}\bar{m}_3(\bfm{Z})\Big{)}.
\]
Now, $\det(\bfm{A}\otimes\bfm{A}\otimes\bfm{A})=\det(\bfm{A}\otimes\bfm{A})^p\det(\bfm{A})^{p^2}=\det(\bfm{A})^{3p^2}>0$, so $\bfm{E}:=(\bfm{A}\otimes\bfm{A}\otimes\bfm{A})$ is a nonsingular matrix such that $\bar{m}_3(\bfm{AZ})=\bfm{E}\bar{m}_3(\bfm{Z})$. Since canonical correlations are invariant under nonsingular linear transformations of the two sets of variables, this means that the canonical correlations between $\bfm{\bar{X}}$ and $\bfm{w_X}$ remain unchanged under the transformation $\bfm{AX}+\bfm{b}$. Thus the canonical correlations between $\bfm{\bar{X}}$ and $\bfm{v_Y}$ must also necessarily remain unchanged. This proves the affine invariance of the statistics.
\end{proof}


\begin{sidewaystable}
\begin{center}
\caption{Population values of statistics (as $n\rightarrow\infty$) for distributions considered in the study.}\label{altpop}
\scriptsize{
\begin{tabular}{|l l | c c | c c c c c | c c c c c |  }
\hline
\textbf{Distribution}& & $\mathbf{b_{1,p}}$& $\mathbf{b_{2,p}}$ & $\mathbf{Z_{2,p}^{(HL)}}$ & $\mathbf{Z_{2,p}^{(W)}}$ & $\mathbf{Z_{2,p}^{(PB)}}$ & $\mathbf{Z_{2,p}^{(max)}}$ & $\mathbf{Z_{2,p}^{(min)}}$ & $\mathbf{Z_{3,p}^{(HL)}}$ & $\mathbf{Z_{3,p}^{(W)}}$ & $\mathbf{Z_{3,p}^{(PB)}}$ & $\mathbf{Z_{3,p}^{(max)}}$ & $\mathbf{Z_{3,p}^{(min)}}$\\\hline
Normal distribution  & $p=2$       & 0 & 8 & 0 & 1 & 0 & 0 & 0 & 0 & 1 & 0 & 0 & 0\\ 
& $p=3$       & 0 & 15         & 0 & 1 & 0 & 0 & 0 & 0 & 1 & 0 & 0 & 0\\ \hline
Indep. $Exp(1)$  & $p=2$       & 8.0 & 20.0          & 1.00 & 0.25 & 2.00 & 0.50 & 0.50 & 0.33 & 0.70 & 0.40 & 0.17 & 0.17\\ 
& $p=3$       & 12.0 & 32.9         & 1.50 & 0.13 & 3.00 & 0.50 & 0.50 & 0.50 & 0.58 & 0.60 & 0.17 & 0.17\\  \hline
$LogN(0,2)$    & $p=2$     & 79.2 & 177.0          & 0.73 & 0.40 & 1.16 & 0.37 & 0.37 & 0.13 & 0.88 & 0.14 & 0.07 & 0.06\\ 
& $p=3$       & 120.0 & 270.7         & 1.05 & 0.28 & 1.61 & 0.36 & 0.34 & 0.17 & 0.84 & 0.18 & 0.07 & 0.05\\  \hline
$LogN(0,1)$    & $p=2$     & 2.8 & 13.6          & 0.60 & 0.49 & 0.87 & 0.32 & 0.28 & 0.17 & 0.83 & 0.19 & 0.09 & 0.08\\ 
& $p=3$       & 4.3 & 24.5         & 0.91 & 0.34 & 1.32 & 0.35 & 0.28 & 0.27 & 0.75 & 0.30 & 0.10 & 0.09\\  \hline
$LogN(0,0.5)$   & $p=2$      & 0.2 & 8.3          & 0.07 & 0.93 & 0.07 & 0.04 & 0.03 & 0.00 & 1.00 & 0.00 & 0.00 & 0.00\\ 
& $p=3$       & 0.2 & 15.5         & 0.11 & 0.90 & 0.11 & 0.05 & 0.03 & 0.01 & 0.99 & 0.01 & 0.00 & 0.00\\  \hline
$Laplace(0,1)$ (type I)  & $p=2$      & 2.7 & 16.0           & 0.67 & 0.44 & 1.00 & 0.33 & 0.33 & 0.30 & 0.72 & 0.36 & 0.15 & 0.15\\ 
& $p=3$       & 6.0 & 28.4          & 1.20 & 0.22 & 2.00 & 0.40 & 0.40 & 0.47 & 0.60 & 0.56 & 0.16 & 0.16\\ \hline
$Laplace(0,1)$ (type II) & $p=2$       & 0.0 & 15.0           & 0.00 & 1.00 & 0.00 & 0.00 & 0.00 & 0.30 & 0.72 & 0.36 & 0.15 & 0.15\\ 
& $p=3$       & 0.0 & 27.0         & 0.00 & 1.00 & 0.00 & 0.00 & 0.00 & 0.47 & 0.60 & 0.56 & 0.16 & 0.16\\ \hline
$Beta(1,1)$    & $p=2$     & 0.1 & 6.0          & 0.10 & 0.90 & 0.11 & 0.07 & 0.03 & 1.38 & 0.09 & 4.71 & 0.75 & 0.63\\ 
& $p=3$       & 0.3 & 12.4         & 0.23 & 0.78 & 0.25 & 0.09 & 0.05 & 2.06 & 0.03 & 7.16 & 0.79 & 0.63\\  \hline
$Beta(1,2)$  & $p=2$       & 0.6 & 7.1          & 0.41 & 0.62 & 0.57 & 0.32 & 0.10 & 0.40 & 0.64 & 0.51 & 0.21 & 0.19\\ 
& $p=3$       & 1.0 & 14.1         & 0.58 & 0.50 & 0.84 & 0.38 & 0.10 & 0.76 & 0.42 & 1.02 & 0.26 & 0.25\\  \hline
$Beta(2,2)$  & $p=2$       & 0.1 & 6.5          & 0.08 & 0.92 & 0.09 & 0.05 & 0.03 & 0.71 & 0.41 & 1.12 & 0.41 & 0.30\\ 
& $p=3$       & 0.3 & 13.1         & 0.18 & 0.83 & 0.20 & 0.07 & 0.05 & 1.07 & 0.26 & 1.74 & 0.47 & 0.30\\  \hline
$\chi_2^2$    & $p=2$     & 7.7 & 23.9          & 0.97 & 0.27 & 1.89 & 0.52 & 0.44 & 0.38 & 0.66 & 0.46 & 0.20 & 0.18\\ 
& $p=3$       & 13.5 & 41.9         & 1.54 & 0.12 & 3.15 & 0.53 & 0.50 & 0.61 & 0.50 & 0.77 & 0.21 & 0.20\\ \hline
$\chi_8^2$    & $p=2$     & 1.9 & 12.0          & 0.54 & 0.53 & 0.76 & 0.32 & 0.22 & 0.16 & 0.85 & 0.18 & 0.09 & 0.07\\ 
& $p=3$       & 3.4 & 21.7          & 0.91 & 0.34 & 1.30 & 0.34 & 0.29 & 0.28 & 0.75 & 0.31 & 0.10 & 0.09\\  \hline
$t(2)$      & $p=2$   & -- & --          & -- & -- & -- & -- & -- & -- & -- & -- & -- & --\\ 
& $p=3$       & -- & --         & -- & -- & -- & -- & -- & -- & -- & -- & -- & --\\ \hline
$AL(\mathbf{0},\bfm{\Sigma_{0}})$   & $p=2$      & 0.0 & 16.0          & 0.00 & 1.00 & 0.00 & 0.00 & 0.00 & 0.33 & 0.69 & 0.40 & 0.17 & 0.17\\ 
& $p=3$       & 0.0 & 29.9         & 0.00 & 1.00 & 0.00 & 0.00 & 0.00 & 0.56 & 0.54 & 0.68 & 0.19 & 0.18\\  \hline
$AL(\mathbf{1},\bfm{\Sigma_{0}})$   & $p=2$      & 5.6 & 20.0          & 0.76 & 0.37 & 1.36 & 0.51 & 0.25 & 0.34 & 0.69 & 0.41 & 0.19 & 0.15\\ 
& $p=3$       & 8.3 & 35.7         & 1.1 & 0.23 & 2.01 & 0.56 & 0.27 & 0.56 & 0.54 & 0.69 & 0.22 & 0.17\\  \hline
$AL(\mathbf{3},\bfm{\Sigma_{0}})$    & $p=2$     & 6.8 & 20.1          & 0.87 & 0.30 & 1.71 & 0.55 & 0.32 & 0.35 & 0.68 & 0.43 & 0.20 & 0.15\\ 
& $p=3$       & 9.8 & 36.7         & 1.25 & 0.18 & 2.43 & 0.60 & 0.33 & 0.57 & 0.53 & 0.70 & 0.23 & 0.17\\ \hline
$AL(\mathbf{1},\bfm{\Sigma_{0.5}})$   & $p=2$      & 5.1 & 19.5          & 0.76 & 0.37 & 1.37 & 0.51 & 0.25 & 0.34 & 0.68 & 0.42 & 0.19 & 0.15\\ 
& $p=3$       & 7.1 & 34.8         & 0.97 & 0.29 & 1.66 & 0.51 & 0.23 & 0.55 & 0.54 & 0.68 & 0.21 & 0.17\\  \hline
$AL(\mathbf{1},\bfm{\Sigma_{0.9}})$   & $p=2$      & 4.7 & 19.3          & 0.66 & 0.43 & 1.11 & 0.46 & 0.20 & 0.34 & 0.69 & 0.40 & 0.18 & 0.15\\ 
& $p=3$       & 6.3 & 34.3         & 0.89 & 0.33 & 1.45 & 0.48 & 0.21 & 0.55 & 0.54 & 0.68 & 0.21 & 0.17\\  \hline
$\frac{9}{10}N(\bfm{0},\bfm{\Sigma_{0}})+\frac{1}{10}N(\bfm{1},\bfm{\Sigma_{0}})$   & $p=2$      & 0.0 & 8.1          & 0.01 & 0.99 & 0.01 & 0.01 & 0.00 & 0.00 & 1.00 & 0.00 & 0.00 & 0.00\\ 
& $p=3$       & 0.1 & 15.2         & 0.03 & 0.97 & 0.03 & 0.03 & 0 & 0.01 & 0.99 & 0.01 & 0.01 & 0.00\\  \hline
$\frac{9}{10}N(\bfm{0},\bfm{\Sigma_{0}})+\frac{1}{10}N(\bfm{2},\bfm{\Sigma_{0}})$    & $p=2$     & 0.5 & 8.9          & 0.18 & 0.82 & 0.22 & 0.18 & 0.00 & 0.07 & 0.93 & 0.07 & 0.07 & 0.00\\ 
& $p=3$       & X & X         & 0.29 & 0.71 & 0.42 & 0.29 & 0 & 0.13 & 0.87 & 0.15 & 0.13 & 0.00\\  \hline
$\frac{9}{10}N(\bfm{0},\bfm{\Sigma_{0}})+\frac{1}{10}N(\bfm{0},\bfm{\Sigma_{0.5}})$  & $p=2$       & 0.0 & 8.1          & 0.00 & 1.00 & 0.00 & 0.00 & 0.00 & 0.00 & 1.00 & 0.00 & 0.00 & 0.00\\ 
& $p=3$       & X & X         & 0.00 & 1.00 & 0.00 & 0.00 & 0.00 & 0.01 & 0.99 & 0.01 & 0.01 & 0.00\\ \hline
$\frac{9}{10}N(\bfm{0},\bfm{\Sigma_{0}})+\frac{1}{10}N(\bfm{1},\bfm{\Sigma_{0.5}})$   & $p=2$      & 0.1 & 8.4          & 0.04 & 0.96 & 0.04 & 0.04 & 0.00 & 0.02 & 0.98 & 0.02 & 0.02 & 0.00\\ 
& $p=3$       & X & X         & 0.10 & 0.90 & 0.11 & 0.09 & 0.01 & 0.08 & 0.92 & 0.08 & 0.06 & 0.01\\  \hline
$\frac{9}{10}N(\bfm{0},\bfm{\Sigma_{0}})+\frac{1}{10}N(\bfm{2},\bfm{\Sigma_{0.5}})$  & $p=2$       & 0.8 & 9.1          & 0.23 & 0.78 & 0.29 & 0.21 & 0.01 & 0.13 & 0.87 & 0.14 & 0.11 & 0.02\\ 
& $p=3$       & X & X         & 0.37 & 0.64 & 0.55 & 0.34 & 0.02 & 0.25 & 0.76 & 0.30 & 0.21 & 0.02\\  \hline
$\frac{3}{4}N(\bfm{0},\bfm{\Sigma_{0}})+\frac{1}{4}N(\bfm{1},\bfm{\Sigma_{0}})$ & $p=2$        & 0.0 & 7.9          & 0.01 & 0.99 & 0.01 & 0.01 & 0.00 & 0.00 & 1.00 & 0.00 & 0.00 & 0.00\\ 
& $p=3$       & X & X         & 0.03 & 0.97 & 0.03 & 0.03 & 0.00 & 0.00 & 1.00 & 0.00 & 0.00 & 0.00\\  \hline
$\frac{3}{4}N(\bfm{0},\bfm{\Sigma_{0}})+\frac{1}{4}N(\bfm{2},\bfm{\Sigma_{0}})$  & $p=2$       & 0.3 & 7.8          & 0.16 & 0.84 & 0.20 & 0.16 & 0.00 & 0.01 & 0.99 & 0.01 & 0.01 & 0.00\\ 
& $p=3$       & X & X         & 0.26 & 0.74 & 0.36 & 0.26 & 0 & 0.03 & 0.97 & 0.03 & 0.03 & 0.00\\  \hline
$\frac{3}{4}N(\bfm{0},\bfm{\Sigma_{0}})+\frac{1}{4}N(\bfm{0},\bfm{\Sigma_{0.5}})$  & $p=2$       & 0.0 & 8.2          & 0.00 & 1.00 & 0.00 & 0.00 & 0.00 & 0.01 & 0.99 & 0.01 & 0.01 & 0.00\\ 
& $p=3$       & X & X         & 0.00 & 1.00 & 0.00 & 0.00 & 0.00 & 0.04 & 0.96 & 0.04 & 0.02 & 0.01\\  \hline
$\frac{3}{4}N(\bfm{0},\bfm{\Sigma_{0}})+\frac{1}{4}N(\bfm{1},\bfm{\Sigma_{0.5}})$  & $p=2$       & 0.2 & 8.2          & 0.08 & 0.92 & 0.08 & 0.06 & 0.02 & 0.02 & 0.98 & 0.02 & 0.01 & 0.01\\ 
& $p=3$       & X & X         & 0.19 & 0.82 & 0.22 & 0.15 & 0.02 & 0.08 & 0.92 & 0.09 & 0.05 & 0.02\\  \hline
$\frac{3}{4}N(\bfm{0},\bfm{\Sigma_{0}})+\frac{1}{4}N(\bfm{2},\bfm{\Sigma_{0.5}})$   & $p=2$      & 0.6 & 7.9          & 0.27 & 0.74 & 0.34 & 0.23 & 0.04 & 0.05 & 0.95 & 0.05 & 0.02 & 0.02\\
& $p=3$       & X & X         & 0.46 & 0.58 & 0.66 & 0.36 & 0.05 & 0.13 & 0.87 & 0.14 & 0.07 & 0.03\\  \hline
\end{tabular}
}
\end{center}
\end{sidewaystable}

\begin{sidewaystable}
\begin{center}
\caption{Power of tests for normality against some alternatives, $\alpha=0.05$, $p=2$.}\label{tab2}
\scriptsize{
\begin{tabular}{|l l| c c c| c c c c c | c c c c c| }
\hline
\bf{Distribution} & & $\mathbf{b_{1,p}}$& $\mathbf{b_{2,p}}$& $\bfm{T}$ & $\mathbf{Z_{2,p}^{(HL)}}$ & $\mathbf{Z_{2,p}^{(W)}}$ & $\mathbf{Z_{2,p}^{(PB)}}$ & $\mathbf{Z_{2,p}^{(max)}}$ & $\mathbf{Z_{2,p}^{(min)}}$ & $\mathbf{Z_{3,p}^{(HL)}}$ & $\mathbf{Z_{3,p}^{(W)}}$ & $\mathbf{Z_{3,p}^{(PB)}}$ & $\mathbf{Z_{3,p}^{(max)}}$ & $\mathbf{Z_{3,p}^{(min)}}$ \\ \hline
Indep. $Exp(1)$ & $n=20$         & 0.79 & 0.54 & 0.55 & \bf 0.84 &\bf 0.86 & \bf 0.86 & \bf 0.84 & 0.72 & 0.24 & 0.23 & 0.20 & 0.15 & 0.26\\ 
				 & $n=50$         &\bf  1.00 & 0.88 & 0.86 &\bf  1.00 &\bf  1.00 &\bf  1.00 &\bf  1.00 &\bf  1.00 & 0.41 & 0.39 & 0.37 & 0.30 & 0.40\\\hline 
$LogN(0,2)$ & $n=20$          &\bf 0.95 & 0.82 & 0.82 &\bf 0.97 &\bf 0.97 &\bf 0.97 &\bf 0.97 & 0.89 & 0.50 & 0.48 & 0.44 & 0.36 & 0.51\\ 
				 & $n=50$          &\bf 1.00 &\bf 0.99 &\bf 0.99 &\bf 1.00 &\bf 1.00 &\bf 1.00 &\bf 1.00 &\bf 1.00 & 0.85 & 0.84 & 0.83 & 0.78 & 0.82\\\hline
$LogN(0,1)$  & $n=20$          &\bf 0.38 & 0.28 & 0.3 & 0.29 & 0.33 & 0.34 & 0.35 & 0.22 & 0.11 & 0.10 & 0.09 & 0.07 & 0.12\\ 
				 & $n=50$          & 0.83 & 0.56 & 0.56 &\bf 0.86 &\bf 0.86 &\bf 0.86 &\bf 0.84 & 0.7 & 0.14 & 0.13 & 0.11 & 0.08 & 0.15\\\hline
$LogN(0,0.5)$  & $n=20$          &\bf 0.07 & 0.06 &\bf 0.07 & 0.06 & 0.06 & 0.06 &\bf 0.07 & 0.06 & 0.05 & 0.05 & 0.05 & 0.05 & 0.06\\ 
				 & $n=50$          &\bf 0.12 & 0.09 & 0.09 & 0.10 & 0.10 & 0.10 &\bf 0.11 & 0.08 & 0.05 & 0.05 & 0.05 & 0.05 & 0.05\\\hline
$Laplace(0,1)$ (type I)  & $n=20$          & 0.57 & 0.51 & 0.44 & 0.58 &\bf 0.60 &\bf 0.60 &\bf 0.59 & 0.50 & 0.06 & 0.05 & 0.04 & 0.03 & 0.07\\
				 & $n=50$          &\bf 0.92 & 0.86 & 0.75 &\bf 0.92 &\bf 0.92 &\bf 0.92 &\bf 0.91 & 0.83 & 0.09 & 0.08 & 0.07 & 0.05 & 0.12\\\hline
$Laplace(0,1)$ (type II)  & $n=20$          & 0.37 &\bf 0.48 & 0.41 & 0.25 & 0.27 & 0.29 & 0.30 & 0.19 & 0.08 & 0.07 & 0.05 & 0.03 & 0.09\\ 
				& $n=50$   & 0.57 &\bf 0.84 & 0.71 & 0.38 & 0.39 & 0.40 & 0.40 & 0.24 & 0.1 & 0.09 & 0.07 & 0.05 & 0.14\\\hline
$Beta(1,1)$   & $n=20$          & 0.02 & 0.00 & 0.01 & 0.05 & 0.06 & 0.06 & 0.06 & 0.04 &\bf 0.41 &\bf 0.43 &\bf 0.43 & 0.38 & 0.34\\ 
				 & $n=50$          & 0.01 & 0.00 & 0.00 & 0.12 & 0.12 & 0.12 & 0.12 & 0.09 &\bf 0.98 &\bf 0.98 &\bf 0.98 &\bf 0.96 & 0.93   \\\hline
$Beta(1,2)$   & $n=20$          & 0.08 & 0.03 & 0.05 & 0.11 & 0.13 & 0.15 & 0.17 & 0.08 &\bf 0.18 &\bf 0.19 &\bf 0.19 &\bf 0.17 & 0.16\\ 
				 & $n=50$          & 0.26 & 0.01 & 0.02 & 0.59 & 0.60 & 0.62 & 0.63 & 0.31 &\bf 0.40 &\bf 0.40 &\bf 0.39 & 0.35 & 0.34   \\\hline
$Beta(2,2)$   & $n=20$          & 0.02 & 0.01 & 0.01 & 0.04 & 0.04 & 0.04 & 0.04 & 0.04 &\bf 0.14 &\bf 0.15 &\bf 0.14 & 0.13 & 0.13\\ 
				 & $n=50$       &   0.02 & 0.00 & 0.00 & 0.08 & 0.08 & 0.08 & 0.07 & 0.09 &\bf 0.53 &\bf 0.52 & 0.50 & 0.43 & 0.46   \\\hline
$\chi_2^2$  & $n=20$          & 0.79 & 0.65 & 0.64 & 0.78 &\bf 0.83 &\bf 0.84 &\bf 0.83 & 0.61 & 0.24 & 0.22 & 0.20 & 0.15 & 0.25\\ 
				 & $n=50$          & \bf 1.00 & 0.95 & 0.92 &\bf  1.00 &\bf  1.00 &\bf  1.00 &\bf  1.00 & 0.98 & 0.36 & 0.35 & 0.33 & 0.27 & 0.34\\ \hline
$\chi_8^2$  & $n=20$          &\bf 0.33 & 0.25 & 0.27 & 0.23 & 0.27 & 0.28 & 0.29 & 0.17 & 0.10 & 0.09 & 0.08 & 0.07 & 0.11\\ 
				 & $n=50$          & 0.76 & 0.51 & 0.50 &\bf  0.79 &\bf  0.80 &\bf  0.80 &\bf  0.78 & 0.59 & 0.11 & 0.10 & 0.09 & 0.07 & 0.12\\ \hline
$t(2)$  & $n=20$          & 0.72 &\bf 0.79 & 0.74 & 0.57 & 0.62 & 0.63 & 0.64 & 0.41 & 0.32 & 0.31 & 0.28 & 0.24 & 0.31\\
				 & $n=50$          & 0.92 &\bf  0.99 &\bf  0.97 & 0.84 & 0.85 & 0.85 & 0.85 & 0.61 & 0.66 & 0.66 & 0.65 & 0.61 & 0.58\\ \hline
$AL(\mathbf{0},\bfm{\Sigma_{0}})$ & $n=20$         & 0.42 &\bf 0.54 & 0.45 & 0.27 & 0.31 & 0.32 & 0.33 & 0.21 & 0.08 & 0.07 & 0.05 & 0.03 & 0.1\\ 
				 & $n=50$         & 0.61 &\bf 0.89 & 0.76 & 0.42 & 0.43 & 0.44 & 0.44 & 0.27 & 0.10 & 0.08 & 0.07 & 0.05 & 0.12\\\hline
$AL(\mathbf{1},\bfm{\Sigma_{0}})$ & $n=20$        &\bf 0.67 & 0.57 & 0.55 & 0.57 & 0.65 &\bf 0.68 &\bf 0.69 & 0.37 & 0.16 & 0.15 & 0.13 & 0.09 & 0.18\\ 
				 & $n=50$          &\bf 0.98 & 0.90 & 0.86 &\bf 0.99 &\bf 0.99 &\bf 0.99 &\bf 0.99 & 0.79 & 0.21 & 0.19 & 0.18 & 0.14 & 0.21\\\hline
%
%
$AL(\mathbf{3},\bfm{\Sigma_{0}})$ & $n=20$         & 0.73 & 0.58 & 0.58 & 0.68 & 0.76 &\bf 0.78 &\bf 0.79 & 0.46 & 0.19 & 0.18 & 0.16 & 0.12 & 0.20\\ 
				 & $n=50$          &\bf 1.00 & 0.91 & 0.88 &\bf 1.00 &\bf 1.00 &\bf 1.00 &\bf 1.00 & 0.92 & 0.24 & 0.23 & 0.21 & 0.17 & 0.24\\\hline
$AL(\mathbf{1},\bfm{\Sigma_{0.5}})$  & $n=20$          &\bf 0.64 & 0.56 & 0.54 & 0.52 & 0.60 &\bf 0.63 &\bf 0.64 & 0.34 & 0.15 & 0.14 & 0.11 & 0.08 & 0.17\\ 
				 & $n=50$          &\bf 0.97 & 0.90 & 0.85 &\bf 0.97 &\bf 0.97 &\bf 0.97 &\bf 0.97 & 0.72 & 0.19 & 0.18 & 0.17 & 0.13 & 0.21\\\hline
$AL(\mathbf{1},\bfm{\Sigma_{0.9}})$   & $n=20$           & \bf 0.61 & 0.57 & 0.53 & 0.50 & \bf 0.60 & \bf 0.59 & \bf 0.61 & 0.33 & 0.14 & 0.13 & 0.11& 0.07 & 0.15\\ 
				 & $n=50$           & \bf 0.95 & 0.90 & 0.85 & \bf 0.95 & \bf 0.95 & \bf 0.96 & \bf 0.96 & 0.68 & 0.19 & 0.18 & 0.16 & 0.13 & 0.21\\  \hline
$\frac{9}{10}N(\bfm{0},\bfm{\Sigma_{0}})+\frac{1}{10}N(\bfm{1},\bfm{\Sigma_{0}})$  & $n=20$          & 0.06 & 0.06 & 0.06 & 0.05 & 0.05 & 0.05 & 0.06 & 0.05 & 0.05 & 0.05 & 0.04 & 0.05 & 0.05\\ 
				 & $n=50$          &\bf 0.07 &\bf 0.07 &\bf 0.07 & 0.06 & 0.06 & 0.06 &\bf 0.07 & 0.06 & 0.05 & 0.04 & 0.04 & 0.04 & 0.05   \\\hline
$\frac{9}{10}N(\bfm{0},\bfm{\Sigma_{0}})+\frac{1}{10}N(\bfm{2},\bfm{\Sigma_{0}})$   & $n=20$          & 0.13 & 0.12 & 0.12 & 0.1 & 0.11 & 0.12 & 0.13 & 0.09 & 0.05 & 0.05 & 0.04 & 0.04 & 0.05\\ 
				 & $n=50$          &\bf 0.33 & 0.19 & 0.17 & 0.27 & 0.29 & 0.30 &\bf 0.33 & 0.13 & 0.05 & 0.04 & 0.04 & 0.04 & 0.05   \\\hline
$\frac{9}{10}N(\bfm{0},\bfm{\Sigma_{0}})+\frac{1}{10}N(\bfm{0},\bfm{\Sigma_{0.5}})$   & $n=20$          & 0.05 & 0.05 &\bf 0.06 & 0.05 & 0.05 & 0.05 & 0.05 & 0.05 & 0.05 & 0.05 & 0.04 & 0.05 & 0.05\\ 
				 & $n=50$          &\bf  0.06 &\bf  0.06 &\bf  0.06 & 0.05 &\bf  0.06 &\bf  0.06 &\bf  0.06 & 0.05 & 0.05 & 0.04 & 0.04 & 0.04 & 0.05\\  \hline
$\frac{9}{10}N(\bfm{0},\bfm{\Sigma_{0}})+\frac{1}{10}N(\bfm{1},\bfm{\Sigma_{0.5}})$   & $n=20$          &\bf 0.07 &\bf 0.07 &\bf 0.07 & 0.06 & 0.06 &\bf 0.07 &\bf 0.07 & 0.06 & 0.05 & 0.05 & 0.04 & 0.05 & 0.05\\ 
				 & $n=50$          &\bf  0.11 &\bf  0.10 &\bf  0.10 & 0.09 &\bf  0.10 &\bf  0.10 &\bf  0.10 & 0.08 & 0.05 & 0.05 & 0.05 & 0.04 & 0.06\\  \hline
$\frac{9}{10}N(\bfm{0},\bfm{\Sigma_{0}})+\frac{1}{10}N(\bfm{2},\bfm{\Sigma_{0.5}})$  & $n=20$          &\bf 0.17 & 0.13 & 0.14 & 0.13 & 0.14 & 0.15 & 0.15 & 0.11 & 0.06 & 0.06 & 0.05 & 0.05 & 0.06\\
				 & $n=50$          &\bf  0.43 & 0.23 & 0.21 & 0.36 & 0.38 & 0.39 & 0.40 & 0.20 & 0.09 & 0.08 & 0.07 & 0.06 & 0.10\\ \hline
$\frac{3}{4}N(\bfm{0},\bfm{\Sigma_{0}})+\frac{1}{4}N(\bfm{1},\bfm{\Sigma_{0}})$  & $n=20$          & 0.05 & 0.05 & 0.05 & 0.05 & 0.05 & 0.05 & 0.05 & 0.05 & 0.05 & 0.05 & 0.05 & 0.05 & 0.05\\ 
				 & $n=50$          &\bf 0.06 & 0.05 & 0.05 &\bf 0.06 &\bf 0.06 &\bf 0.06 &\bf 0.06 & 0.05 & 0.05 & 0.05 & 0.05 & 0.05 & 0.05\\ \hline
$\frac{3}{4}N(\bfm{0},\bfm{\Sigma_{0}})+\frac{1}{4}N(\bfm{2},\bfm{\Sigma_{0}})$   & $n=20$          & 0.07 & 0.05 & 0.05 & 0.08 &0.08 &\bf 0.09 &\bf 0.10 & 0.07 & 0.05 & 0.05 & 0.05 & 0.05 & 0.05\\ 
				 & $n=50$          & 0.14 & 0.04 & 0.04 & 0.21 & 0.23 & 0.24 &\bf 0.27 & 0.12 & 0.06 & 0.06 & 0.06 & 0.06 & 0.06\\ \hline
$\frac{3}{4}N(\bfm{0},\bfm{\Sigma_{0}})+\frac{1}{4}N(\bfm{0},\bfm{\Sigma_{0.5}})$   & $n=20$          &\bf 0.06 &\bf 0.06 &\bf 0.06 & 0.05 & 0.05 &\bf 0.06 &\bf 0.06 & 0.05 & 0.05 & 0.05 & 0.04 & 0.05 & 0.05\\ 
				 & $n=50$          & 0.07 &\bf 0.08 & 0.07 & 0.06 & 0.06 & 0.06 & 0.06 & 0.06 & 0.05 & 0.04 & 0.04 & 0.04 & 0.05\\ \hline
$\frac{3}{4}N(\bfm{0},\bfm{\Sigma_{0}})+\frac{1}{4}N(\bfm{1},\bfm{\Sigma_{0.5}})$   & $n=20$          & 0.07 & 0.07 & 0.07 & 0.07 &\bf 0.08 &\bf 0.08 &\bf 0.08 & 0.07 & 0.05 & 0.05 & 0.04 & 0.05 & 0.05\\ 
				 & $n=50$          &\bf 0.13 & 0.08 & 0.08 &\bf 0.14 &\bf 0.14 &\bf 0.14 &\bf 0.13 & 0.12 & 0.05 & 0.05 & 0.05 & 0.04 & 0.06  \\ \hline
$\frac{3}{4}N(\bfm{0},\bfm{\Sigma_{0}})+\frac{1}{4}N(\bfm{2},\bfm{\Sigma_{0.5}})$  & $n=20$          & 0.11 & 0.06 & 0.07 &\bf 0.14 &\bf 0.15 &\bf 0.15 &\bf 0.15 &\bf 0.14 & 0.05 & 0.05 & 0.05 & 0.05 & 0.05\\
				 & $n=50$      & 0.32 & 0.06 & 0.06 &\bf 0.44 &\bf 0.45 &\bf 0.45 &\bf 0.44 & 0.31 & 0.07 & 0.07 & 0.06 & 0.06 & 0.08    \\ \hline
\end{tabular}
}
\end{center}
\end{sidewaystable}



\begin{sidewaystable}
\begin{center}
\caption{Power of tests for normality against some alternatives, $\alpha=0.05$, $p=3$.}\label{tab4}
\scriptsize{
\begin{tabular}{|l l| c c c| c c c c c | c c c c c| }
\hline
\bf{Distribution} & & $\mathbf{b_{1,p}}$& $\mathbf{b_{2,p}}$& $\bfm{T}$ & $\mathbf{Z_{2,p}^{(HL)}}$ & $\mathbf{Z_{2,p}^{(W)}}$ & $\mathbf{Z_{2,p}^{(PB)}}$ & $\mathbf{Z_{2,p}^{(max)}}$ & $\mathbf{Z_{2,p}^{(min)}}$ & $\mathbf{Z_{3,p}^{(HL)}}$ & $\mathbf{Z_{3,p}^{(W)}}$ & $\mathbf{Z_{3,p}^{(PB)}}$ & $\mathbf{Z_{3,p}^{(max)}}$ & $\mathbf{Z_{3,p}^{(min)}}$ \\ \hline
Indep. $Exp(1)$ & $n=20$          & 0.82 & 0.61 & 0.62 & \bf 0.84 & \bf 0.86 & \bf 0.84 & 0.77 & 0.68 & 0.29 & 0.29 & 0.26 & 0.19 & 0.26   \\
				 & $n=50$          & \bf 1.00 & 0.93 & 0.91 & \bf 1.00 & \bf 1.00 & \bf 1.00 & \bf 1.00 & \bf 1.00 & 0.59 & 0.59 & 0.57 & 0.46 & 0.49   \\\hline
$LogN(0,2)$ & $n=20$          & \bf 0.97 & 0.89 & 0.88 & \bf 0.98 & \bf 0.98 & \bf 0.98 & \bf 0.96 & 0.88 & 0.61 & 0.63 & 0.62 & 0.52 & 0.53   \\
				 & $n=50$          & \bf 1.00 & \bf 1.00 & \bf 1.00 & \bf 1.00 & \bf 1.00 & \bf 1.00 & \bf 1.00 & \bf 1.00 & 0.95 & 0.95 & 0.95 & 0.90 & 0.88   \\\hline
$LogN(0,1)$  & $n=20$          & \bf 0.41 & 0.33 & 0.36 & 0.28 & 0.30 & 0.30 & 0.29 & 0.21 & 0.13 & 0.12 & 0.11 & 0.08 & 0.13   \\
				 & $n=50$          & \bf 0.89 & 0.69 & 0.67 & \bf 0.87 & \bf 0.88 & \bf 0.88 & 0.84 & 0.63 & 0.22 & 0.21 & 0.19 & 0.14 & 0.21   \\\hline
$LogN(0,0.5)$  & $n=20$          & \bf 0.07 & \bf 0.07 & \bf 0.07 & 0.06 & 0.06 & 0.06 & 0.06 & 0.06 & 0.05 & 0.05 & 0.05 & 0.05 & 0.05   \\
				 & $n=50$          & \bf 0.11 & 0.09 & 0.09 & 0.09 & 0.09 & 0.09 & 0.09 & 0.08 & 0.06 & 0.06 & 0.06 & 0.06 & 0.06  \\\hline
$Laplace(0,1)$ (type I)  & $n=20$          & 0.71 & 0.61 & 0.53 & \bf 0.75 & \bf 0.76 &\bf  0.75 & 0.68 & 0.59 & 0.09 & 0.08 & 0.06 & 0.04 & 0.11  \\
				 & $n=50$          & \bf 0.99 & 0.94 & 0.86 & \bf 0.99 & \bf 1.00 & \bf 0.99 & \bf 0.99 & 0.94 & 0.24 & 0.23 & 0.20 & 0.12 & 0.24   \\\hline
$Laplace(0,1)$ (type II)  & $n=20$          & 0.46 & \bf 0.58 & 0.49 & 0.33 & 0.35 & 0.35 & 0.32 & 0.23 & 0.09 & 0.08 & 0.07 & 0.05 & 0.10   \\
				 & $n=50$          & 0.72 & \bf 0.93 & 0.83 & 0.51 & 0.53 & 0.54 & 0.52 & 0.30 & 0.17 & 0.16 & 0.14 & 0.08 & 0.16   \\\hline
$Beta(1,1)$   & $n=20$          & 0.02 & 0.01 & 0.02 & 0.06 & 0.06 & 0.06 & 0.06 & 0.05 & 0.28 & \bf 0.31 & \bf 0.33 & 0.28 & 0.20    \\
				 & $n=50$          & 0.02 & 0.00 & 0.00 & 0.15 & 0.16 & 0.16 & 0.16 & 0.11 & \bf 0.96 & \bf 0.97 & \bf 0.96 & 0.92 & 0.85  \\\hline
$Beta(1,2)$   & $n=20$          & 0.08 & 0.05 & 0.06 & 0.09 & 0.10 & 0.10 & 0.11 & 0.09 & \bf 0.16 & \bf 0.16 & \bf 0.17 & 0.14 & 0.13   \\
				 & $n=50$          & 0.21 & 0.02 & 0.03 & 0.45 & 0.49 & 0.52 & \bf 0.56 & 0.20 & 0.45 & 0.46 & 0.46 & 0.38 & 0.35   \\\hline
$Beta(2,2)$   & $n=20$          & 0.02 & 0.01 & 0.02 & 0.04 & 0.05 & 0.05 & 0.05 & 0.04 & \bf 0.10 & \bf 0.10 & \bf 0.11 & \bf 0.10 & 0.08   \\
				 & $n=50$         & 0.02 & 0.00 & 0.01 & 0.10 & 0.10 & 0.09 & 0.09 & 0.09 & \bf 0.40 & \bf 0.40 & \bf 0.39 & 0.32 & 0.32  \\\hline
$\chi_2^2$  & $n=20$          & \bf 0.89 & 0.79 & 0.76 & \bf 0.88 & \bf 0.90 & \bf 0.90 & 0.86 & 0.72 & 0.35 & 0.34 & 0.32 & 0.23 & 0.31   \\
				 & $n=50$          & \bf 1.00 & \bf 0.99 & 0.97 & \bf 1.00 & \bf 1.00 & \bf 1.00 & \bf 1.00 & \bf 0.99 & 0.71 & 0.71 & 0.69 & 0.56 & 0.60   \\\hline
$\chi_8^2$  & $n=20$          & \bf 0.39 & 0.31 & 0.32 & 0.31 & 0.33 & 0.32 & 0.29 & 0.23 & 0.10 & 0.09 & 0.08 & 0.06 & 0.10  \\
				 & $n=50$          & \bf 0.87 & 0.62 & 0.59 & \bf 0.87 & \bf 0.88 & \bf 0.87 & 0.82 & 0.68 & 0.18 & 0.18 & 0.16 & 0.12 & 0.17   \\\hline
$t(2)$  & $n=20$          & 0.83 & \bf 0.89 & 0.85 & 0.70 & 0.75 & 0.76 & 0.74 & 0.50 & 0.34 & 0.36 & 0.36 & 0.32 & 0.30   \\
				 & $n=50$          & \bf 0.98 & \bf 1.00 & \bf 0.99 & 0.95 & 0.95 & 0.95 & 0.95 & 0.72 & 0.77 & 0.77 & 0.75 & 0.69 & 0.69   \\\hline	
$AL(\mathbf{0},\bfm{\Sigma_{0}})$ & $n=20$ & 0.58  & \bf 0.71 & 0.60 & 0.43 & 0.46 & 0.46 & 0.44 & 0.30 & 0.10 & 0.08 & 0.07 & 0.04 & 0.11            \\
				 & $n=50$          & 0.80 & \bf 0.98 & 0.90 & 0.61 & 0.62 & 0.63 & 0.62 & 0.36 & 0.16 & 0.14 & 0.13 & 0.08 & 0.15  \\\hline
$AL(\mathbf{1},\bfm{\Sigma_{0}})$ & $n=20$          & \bf 0.75 & 0.71 & 0.68 & 0.58 & 0.65 & 0.67 & 0.66 & 0.38 & 0.24 & 0.22 & 0.21 & 0.15 & 0.22  \\
				 & $n=50$          & \bf 0.99 & \bf 0.98 & 0.95 & \bf 0.99 & \bf 0.99 & \bf 0.99 & \bf 0.99 & 0.73 & 0.35 & 0.34 & 0.32 & 0.23 & 0.29  \\\hline
$AL(\mathbf{3},\bfm{\Sigma_{0}})$  & $n=20$          & \bf 0.78 & 0.71 & 0.70 & 0.64 & 0.72 & 0.74 & 0.73 & 0.41 & 0.28 & 0.27 & 0.26 & 0.19 & 0.26   \\
				 & $n=50$          & \bf 1.00 & \bf 0.98 & 0.96 & \bf 1.00 & \bf 1.00 & \bf 1.00 & \bf 1.00 & 0.83 & 0.39 & 0.39 & 0.37 & 0.28 & 0.32   \\\hline
$AL(\mathbf{1},\bfm{\Sigma_{0.5}})$  & $n=20$          & \bf 0.72 & \bf 0.71 & 0.66 & 0.55 & 0.61 & 0.62 & 0.61 & 0.36 & 0.20 & 0.19 & 0.17 & 0.12 & 0.20   \\
				 & $n=50$          & \bf 0.98 & \bf 0.98 & 0.94 & \bf 0.97 & \bf 0.98 & \bf 0.98 & \bf 0.98 & 0.65 & 0.31 & 0.30 & 0.27 & 0.19 & 0.27  \\\hline
$AL(\mathbf{1},\bfm{\Sigma_{0.9}})$   & $n=20$          & \bf 0.70 & \bf 0.71 & 0.66 & 0.53 & 0.58 & 0.60 & 0.58 & 0.35 & 0.19 & 0.17 & 0.16 & 0.11 & 0.18   \\
				 & $n=50$          & \bf 0.98 & \bf 0.98 & 0.94 & 0.95 & \bf 0.96 & \bf 0.96 & \bf 0.96 & 0.61 & 0.29 & 0.28 & 0.25 & 0.17 & 0.25  \\\hline
$\frac{9}{10}N(\bfm{0},\bfm{\Sigma_{0}})+\frac{1}{10}N(\bfm{1},\bfm{\Sigma_{0}})$  & $n=20$          & \bf 0.06 & \bf  0.06 & \bf  0.06 & 0.05 & 0.05 & 0.05 & 0.05 & 0.05 & 0.05 & 0.05 & 0.05 & 0.05 & 0.05   \\
				 & $n=50$          & \bf 0.07 & \bf 0.07 & \bf 0.07 & 0.06 & \bf 0.07 & \bf 0.07 & \bf 0.07 & 0.06 & 0.05 & 0.05 & 0.05 & 0.05 & 0.05   \\\hline
$\frac{9}{10}N(\bfm{0},\bfm{\Sigma_{0}})+\frac{1}{10}N(\bfm{2},\bfm{\Sigma_{0}})$   & $n=20$          & \bf 0.15 & 0.13 & 0.13 & 0.11 & 0.12 & 0.12 & 0.12 & 0.09 & 0.06 & 0.06 & 0.06 & 0.05 & 0.06   \\
				 & $n=50$          & \bf 0.39 & 0.21 & 0.19 & 0.27 & 0.32 & 0.36 & \bf 0.42 & 0.12 & 0.07 & 0.07 & 0.07 & 0.06 & 0.07   \\\hline
$\frac{9}{10}N(\bfm{0},\bfm{\Sigma_{0}})+\frac{1}{10}N(\bfm{0},\bfm{\Sigma_{0.5}})$   & $n=20$          & 0.05 & \bf 0.06 &  \bf 0.06 & 0.05 & 0.05 & 0.05 & 0.05 & 0.05 & 0.05 & 0.05 & 0.05 & 0.05 & 0.05   \\
				 & $n=50$          & 0.06 & \bf 0.07 & \bf 0.07 & 0.06 & 0.06 & 0.06 & 0.06 & 0.06 & 0.05 & 0.05 & 0.05 & 0.05 & 0.05   \\\hline
$\frac{9}{10}N(\bfm{0},\bfm{\Sigma_{0}})+\frac{1}{10}N(\bfm{1},\bfm{\Sigma_{0.5}})$   & $n=20$          & 0\bf .16 & 0.11 & 0.11 & 0.14 & 0.14 & 0.14 & 0.14 & 0.12 & 0.05 & 0.05 & 0.05 & 0.05 & 0.06   \\
				 & $n=50$          & \bf 0.15 & 0.12 & 0.12 & 0.11 & 0.12 & 0.12 & 0.12 & 0.09 & 0.07 & 0.07 & 0.06 & 0.05 & 0.07   \\\hline
$\frac{9}{10}N(\bfm{0},\bfm{\Sigma_{0}})+\frac{1}{10}N(\bfm{2},\bfm{\Sigma_{0.5}})$  & $n=20$          & \bf 0.74 & 0.32 & 0.27 & 0.68 & 0.69 & 0.70 & \bf 0.71 & 0.30 & 0.14 & 0.13 & 0.12 & 0.09 & 0.18   \\
				 & $n=50$          & \bf 0.54 & 0.25 & 0.25 & 0.40 & 0.44 & 0.48 & \bf 0.51 & 0.16 & 0.14 & 0.14 & 0.13 & 0.09 & 0.13   \\\hline
$\frac{3}{4}N(\bfm{0},\bfm{\Sigma_{0}})+\frac{1}{4}N(\bfm{1},\bfm{\Sigma_{0}})$  & $n=20$          & 0.05 & 0.05 & 0.05 & 0.05 & 0.05 & 0.05 & 0.05 & 0.05 & 0.05 & 0.05 & 0.05 & 0.05 & 0.05   \\
				 & $n=50$          & 0.05 & 0.05 & 0.05 & \bf 0.06 & \bf 0.06 & \bf 0.06 & \bf 0.06 & \bf 0.06 & 0.05 & 0.05 & 0.05 & 0.05 & 0.05   \\\hline
$\frac{3}{4}N(\bfm{0},\bfm{\Sigma_{0}})+\frac{1}{4}N(\bfm{2},\bfm{\Sigma_{0}})$   & $n=20$          & \bf 0.07 & 0.05 & 0.05 & \bf 0.09 & \bf 0.09 & \bf 0.09 & \bf 0.09 & \bf 0.08 & 0.06 & 0.06 & 0.06 & 0.06 & 0.05   \\
				 & $n=50$          & 0.12 & 0.04 & 0.04 & 0.21 & 0.24 & 0.27 &\bf 0.32 & 0.11 & 0.06 & 0.06 & 0.06 & 0.06 & 0.06   \\\hline
$\frac{3}{4}N(\bfm{0},\bfm{\Sigma_{0}})+\frac{1}{4}N(\bfm{0},\bfm{\Sigma_{0.5}})$   & $n=20$          & 0.06 & 0.07 & 0.07 & 0.06 & 0.06 & 0.06 & 0.06 & 0.06 & 0.05 & 0.05 & 0.05 & 0.05 & 0.05   \\
				 & $n=50$          & \bf 0.08 & \bf 0.09 & \bf 0.09 & 0.06 & 0.07 & 0.07 & 0.06 & 0.06 & 0.05 & 0.05 & 0.05 & 0.05 & 0.06  \\\hline
$\frac{3}{4}N(\bfm{0},\bfm{\Sigma_{0}})+\frac{1}{4}N(\bfm{1},\bfm{\Sigma_{0.5}})$   & $n=20$          & \bf 0.08 & 0.07 & 0.07 & \bf 0.09 & \bf 0.08 & \bf 0.08 & \bf 0.08 & \bf 0.08 & 0.05 & 0.05 & 0.05 & 0.05 & 0.05   \\
				 & $n=50$          & 0.18 & 0.10 & 0.09 & \bf 0.19 & \bf 0.19 & \bf 0.20 & 0.17 & 0.14 & 0.06 & 0.06 & 0.06 & 0.05 & 0.07   \\\hline
$\frac{3}{4}N(\bfm{0},\bfm{\Sigma_{0}})+\frac{1}{4}N(\bfm{2},\bfm{\Sigma_{0.5}})$  & $n=20$          & 0.11 & 0.07 & 0.07 & \bf 0.15 &\bf  0.15 &\bf  0.15 & 0.13 & 0.13 & 0.06 & 0.06 & 0.06 & 0.06 & 0.06   \\
				 & $n=50$          & 0.35 & 0.06 & 0.06 & 0.51 & 0.54 & \bf 0.57 & \bf 0.56 & 0.25 & 0.08 & 0.08 & 0.08 & 0.07 & 0.08   \\\hline 
\end{tabular}
}
\end{center}
\end{sidewaystable}

\pagebreak

\end{document}